\documentclass[11pt,letterpaper]{article}
\usepackage{amsmath,amssymb,amsfonts,amsthm}
\usepackage{tikz}
\usepackage{tikz-cd}
\usepackage{mathrsfs}
\usepackage[mathscr]{euscript}
\usepackage{dsfont}
\usepackage{enumitem}
\usepackage{bbm}
\usepackage{xcolor}
\usepackage[ruled,vlined,linesnumbered]{algorithm2e}
\usepackage[margin=1in]{geometry}
\usepackage[colorlinks, citecolor = blue, linkcolor = magenta]{hyperref}
\usepackage[noabbrev,capitalise,nameinlink]{cleveref}
\Crefname{corollary}{Corollary}{Corollaries}
\numberwithin{equation}{section}

\allowdisplaybreaks

\newtheorem{theorem}{Theorem}[section]
\newtheorem{lemma}[theorem]{Lemma}
\newtheorem{corollary}[theorem]{Corollary}
\newtheorem{question}[theorem]{Question}

\theoremstyle{definition}
\newtheorem{remark}[theorem]{Remark}
\newtheorem{proposition}[theorem]{Proposition}
\newtheorem{definition}[theorem]{Definition}

\renewcommand{\leq}{\leqslant}

\renewcommand{\geq}{\geqslant}
\renewcommand{\ge}{\geqslant}

\newcommand{\CSP}{\mathrm{CSP}}
\newcommand{\Pol}{\mathrm{Pol}}

\newcommand{\val}{\mathrm{val}}
\newcommand{\supp}{\mathrm{supp}}

\newcommand{\orig}{\mathrm{orig}}
\newcommand{\const}{\mathrm{const}}
\newcommand{\aux}{\mathrm{aux}}

\newcommand{\TLin}{\textsc{3Sum}}
\newcommand{\Const}{\textsc{Const}}

\newcommand{\End}{\textit{End}}
\newcommand{\Perm}{\textit{Perm}}

\newcommand{\BD}{\mathrm{BD}}

\newcommand{\yes}{\mathrm{yes}}
\newcommand{\no}{\mathrm{no}}
\newcommand{\Dno}{\mathcal{D}_{\mathrm{no}}}
\newcommand{\Dyes}{\mathcal{D}_{\mathrm{yes}}}
\newcommand{\Dpair}{\mathcal{D}_{\mathrm{pair}}}

\newcommand{\Mcsp}[3]{\mathsf{MaxCSP}(#1)[#2, #3]}

\newcommand{\Sub}{\mathrm{Sub}}

\newcommand{\Res}{\mathrm{Res}}

\newcommand{\clo}[1]{\langle #1 \rangle^{*}}

\newcommand{\cA}{\ensuremath{\mathcal{A}}}

\newcommand{\cC}{\ensuremath{\mathcal{C}}}

\newcommand{\cE}{\ensuremath{\mathcal{E}}}
\newcommand{\cF}{\ensuremath{\mathcal{F}}}

\newcommand{\cH}{\ensuremath{\mathcal{H}}}
\newcommand{\cI}{\ensuremath{\mathcal{I}}}

\newcommand{\cM}{\ensuremath{\mathcal{M}}}

\newcommand{\cQ}{\ensuremath{\mathcal{Q}}}

\newcommand{\cS}{\ensuremath{\mathcal{S}}}
\newcommand{\cT}{\ensuremath{\mathcal{T}}}

\newcommand{\cV}{\ensuremath{\mathcal{V}}}

\newcommand{\bA}{\ensuremath{\mathbb{A}}}

\newcommand{\bE}{\ensuremath{\mathbb{E}}}

\newcommand{\bP}{\ensuremath{\mathbb{P}}}

\newcommand{\bZ}{\ensuremath{\mathbb{Z}}}

\newcommand{\bfv}{\ensuremath{\mathbf{v}}}

\newcommand{\Ex}[1]{\bE \left[ #1 \right]}
\newcommand{\Exu}[2]{\underset{#1} \bE \left[ #2 \right] }

\renewcommand{\Pr}[1]{\bP \left[ #1 \right]} 
\newcommand{\Pru}[2]{\underset{ #1 }\bP \left[ #2 \right]}

\title{Unbounded-width CSPs are Untestable in a Sublinear Number of Queries}
\author{Yumou Fei\thanks{Department of EECS, Massachusetts Institute of Technology.}}
\date{\vspace{-5ex}}

\begin{document}

\maketitle 

\begin{abstract}
The bounded-degree query model, introduced by Goldreich and Ron (\textit{Algorithmica, 2002}), is a standard framework in graph property testing and sublinear-time algorithms. Many properties studied in this model, such as bipartiteness and 3-colorability of graphs, can be expressed as satisfiability of constraint satisfaction problems (CSPs). We prove that for the entire class of \emph{unbounded-width} CSPs, testing satisfiability requires $\Omega(n)$ queries in the bounded-degree model. This result unifies and generalizes several previous lower bounds. In particular, it applies to all CSPs that are known to be $\mathbf{NP}$-hard to solve, including $k$-colorability of $\ell$-uniform hypergraphs for any $k,\ell \ge 2$ with $(k,\ell) \neq (2,2)$. 

Our proof combines the techniques from Bogdanov, Obata, and Trevisan (\textit{FOCS, 2002}), who established the first $\Omega(n)$ query lower bound for CSP testing in the bounded-degree model, with known results from universal algebra.
\end{abstract}

\newpage 
\tableofcontents

\newpage 
\section{Introduction}

Property testing is a class of algorithmic problems in which the algorithm is allowed to inspect only a small portion of the input. In exchange for this limited access, the goal is relaxed: instead of deciding whether the input exactly satisfies a given property, the tester must merely distinguish inputs that satisfy the property from those that are \emph{far} from satisfying it, under an appropriate notion of distance. 

In this paper, we consider the problem of testing satisfiability of constraint satisfaction problems (CSPs). We formally define CSPs as follows. First of all, every CSP considered in this paper is specified by a finite \emph{template}.

\begin{definition}\label{def:relational_structure}
A \emph{relational structure} is a pair $(D,\Gamma)$ where $D$ is a finite set, and $\Gamma$ is a finite collection of relations on $D$. Each relation $R \in \Gamma$ is a function $R : D^k \to \{0,1\}$ for some positive integer $k$ (called the \emph{arity} of $R$). A fixed relational structure $(D,\Gamma)$ is also referred to as a \emph{CSP template}.
\end{definition}

As a running example, if we let $D=\{0,1\}$ and take $\Gamma=\{\mathds{1}_{\{10,01\}}\}$, where $\mathds{1}_{\{10,01\}}(x,y)=1$ if and only if $x\neq y$, then the template $(D,\Gamma)$ could express the bipartiteness property of graphs, as we will soon explain. 

Given a CSP template $(D,\Gamma)$, we can collect all \emph{instances} of the CSP into an infinite set $\CSP(\Gamma)$.

\begin{definition}\label{def:CSP_instance}
Given a CSP template $(D,\Gamma)$, we define $\CSP(\Gamma)$ as the (infinite) collection of all \emph{instances} $I=(V,\cC)$, where $V$ is a finite set of \emph{variables} and $\cC=\{C_{1},\dots,C_{m}\}$ is a multi-set of \emph{constraints} satisfying the following. Each constraint $C_i$ is a pair $(\bfv, R)$, where $\bfv = (v_1, \dots, v_k)$ is a tuple of distinct\footnote{Some authors allow repetition of variables in a constraint. In our setting, one can also ``artificially'' allow repetitions by adding more relations to the relational structure. See \Cref{subsec:relational_structures} for details.} variables from $V$, and $R : D^k \to \{0,1\}$ is a relation from $\Gamma$. The variable set $\{v_{1},\dots,v_{k}\}\subseteq V$ is called the \emph{scope} of $C_{i}$. 
\end{definition}

Now an instance in $\CSP(\{\mathds{1}_{\{10,01\}}\})$ can be viewed as a graph, with variables as vertices and constraints as edges. The graph is bipartite if and only if the CSP instance is ``satisfiable,'' as defined in the following definition.

\begin{definition}\label{def:CSP_value}
An \emph{assignment} for an instance $I=(V,\cC)\in\CSP(\Gamma)$ is a map $\tau:V\rightarrow D$. We say $\tau$ \emph{satisfies} a constraint $((v_{1},\dots,v_{k}),R)$ if $R\big(\tau(v_{1}),\dots,\tau(v_{k})\big)=1$. The \emph{value} of an assignment $\tau$, denoted by $\val_{I}(\tau)$, is the fraction of constraints in $\cC$ satisfied by $\tau$. Note that $\val_{I}(\tau)\in [0,1]$. If $\val_{I}(\tau)=1$, we say that $\tau$ \emph{satisfies} the instance $I$, or $\tau$ is a \emph{satisfying assignment} of $I$. The \emph{value} of the instance $I$, denoted by $\val_{I}$, is the maximum possible value any assignment may achieve. 
\end{definition}

The CSP defined by the bipartiteness template $\big(\{0,1\},\{\mathds{1}_{\{10,01\}}\}\big)$ will be denoted by $\mathsf{2COL}$. Similarly, the property of vertex 3-colorability in graphs can be expressed as a constraint satisfaction problem denoted by $\mathsf{3COL}$.

\subsection{Testing CSP Satisfiability}

We study the property testing of CSP satisfiability in the \emph{bounded-degree query model}, defined in \Cref{def:BD_model} below. In this model, the input CSP instance has a fixed variable set, and the degree of each variable (i.e. the number of constraints involving the variable) in the instance is assumed to be uniformly bounded by a constant.

\begin{definition}
Let $(D,\Gamma)$ be a CSP template and let $V$ be a set of variables. We write $\CSP(\Gamma,V)\subseteq \CSP(\Gamma)$ for the collection of all $\CSP(\Gamma)$-instances whose variable set is $V$. 
\end{definition}
\begin{definition}\label{def:BD_model}
Fix a CSP template $(D,\Gamma)$. In the $\BD(d,n)$ model for testing satisfiability, the tester is given oracle access to a CSP instance $I=([n],\cC)\in \CSP(\Gamma,[n])$ with maximum degree at most $d$. The goal of an $\varepsilon$-tester, where $\varepsilon>0$ is a fixed constant, is to achieve the following:
\begin{enumerate}[label=(\arabic*)]
\item If $I$ is satisfiable, the tester must accept $I$ with probability at least $2/3$.
\item If at least $\varepsilon dn$ constraints must be removed from $\cC$ to make $I$ satisfiable, the tester must reject $I$ with probability at least $2/3$. 
\end{enumerate}
The tester accesses the instance $I$ through \emph{oracle queries}: when the tester queries a variable $v\in [n]$, the oracle reveals an arbitrary one of the previously unseen constraints involving $v$, unless all constraints involving $v$ have been revealed (in which case the oracle returns $\perp$).
\end{definition}

The bounded-degree query model was first introduced by \cite{goldreich2002property}, in the context of graph property testing. One of the earliest results in this model is that testing bipartiteness (i.e. testing satisfiability of $\mathsf{2COL}$ instances) requires at least $\Omega(\sqrt{n})$ queries~\cite{goldreich2002property} (for some degree bound $d$ and error $\varepsilon>0$) and can be done with at most $\widetilde{O}(\sqrt{n})$ queries~\cite{goldreich1999sublinear} (for any degree bound $d$ and any error $\varepsilon>0$). In contrast, Bogdanov, Obata, and Trevisan~\cite{bogdanov2002lower} showed that testing vertex 3-colorability requires $\Omega(n)$ queries.\footnote{These seminal works, along with a large number of subsequent studies (such as~\cite{goldreich2011testing,benjaminia2010every,yoshida2010query,yoshida2012property,czumaj2019planar}, to name a few), also investigated a wide range of other graph properties within the bounded-degree model, many of which are not constraint satisfaction problems.} 


The study of testing CSP satisfiability in the bounded-degree model was already formally initiated in the early work of Bogdanov, Obata, and Trevisan~\cite{bogdanov2002lower}, who showed that, beyond the specific case of $\mathsf{3COL}$, several other CSPs (such as $\mathsf{3SAT}$ and $\mathsf{3LIN}$) require $\Omega(n)$ queries to test in this model. A more systematic investigation was later carried out by Yoshida~\cite{yoshida2011optimal}, who proved that any CSP \emph{robustly solved by BasicLP} (the basic linear program) has a satisfiability tester that makes a constant number of queries. Typical examples of such CSPs include $\mathsf{1SAT}$ and $\mathsf{3HornSAT}$. For the satisfiability testing of any CSP that is not \emph{solved by BasicLP},~\cite{yoshida2011optimal} further showed that at least $\Omega(\sqrt{n})$ queries are necessary. Examples of such CSPs include $\mathsf{2SAT}$ and $\mathsf{2COL}$.\footnote{Robust solvability by BasicLP is a stronger condition than solvability by BasicLP, so \cite{yoshida2011optimal} left open the possibility of CSPs with superconstant but $o(\sqrt{n})$ query complexity. However, a recent unpublished note by Brady~\cite{brady2022notes} proves that all CSPs solvable by BasicLP are also robustly solvable by it, thereby turning the result of \cite{yoshida2011optimal} into a full dichotomy theorem.}

The main result of this paper strengthens this picture by proving that, among these BasicLP-unsolvable CSPs, the subclass of \emph{unbounded-width} CSPs requires $\Omega(n)$ queries to test.

\begin{theorem}\label{thm:main}
For any CSP template $(D,\Gamma)$ of unbounded width, there exist constants $\varepsilon\in (0,1)$ and $d\in \mathbb{N}$ such that any $\varepsilon$-tester for satisfiability of $\CSP(\Gamma)$ instances in the $\BD(d,n)$ model must make $\Omega_{d,\varepsilon}(n)$ queries.
\end{theorem}

The notion of bounded width (formally defined in \Cref{sec:bounded-width}) originates from~\cite{feder1998computational} and has been extensively studied in the universal-algebraic approach to CSPs. Bounded-width CSPs can be solved in polynomial time (see \Cref{prop:bounded-width_in_P}), so the class of unbounded-width CSPs includes all CSPs whose satisfiability is $\textbf{NP}$-hard to decide (unless $\textbf{P}=\textbf{NP}$, in which case all CSPs are trivially $\textbf{NP}$-hard). Consequently, \Cref{thm:main} implies that all CSPs known to be $\textbf{NP}$-hard are also unconditionally (and maximally) hard to test in the bounded-degree model. This encompasses, for example, the property of $k$-colorability in $\ell$-uniform hypergraphs for all $k,\ell \ge 2$ with $(k,\ell) \neq (2,2)$ (see \Cref{rem:colorability_unbounded_width}). Our result thus unifies and extends several prior lower bounds from works such as~\cite{bogdanov2002lower,yoshida2010query,aaronson2025property}.\footnote{The recent work of \cite{aaronson2025property} proves that testing $k$-colorability in $k$-uniform hypergraphs requires $\Omega(n)$ queries, for any $k\geq 3$.}

\subsection{Discussion: Testing vs. Approximation}

A property tester can be viewed as an algorithm that approximates the distance of an input to satisfying a property --- albeit in a very coarse sense, as it only distinguishes zero distance from sufficiently large distance. Nevertheless, many property testers (for example, the constant-query tester of~\cite{yoshida2011optimal} for certain CSPs) operate by explicitly estimating this distance and checking whether it is close to zero. On the negative side, a hardness result for property testing (such as our \Cref{thm:main}) can often be interpreted as a hardness of approximation result. This connection to approximation algorithms was, in fact, one of the original motivations in~\cite{goldreich1998property}, where graph property testing was first introduced. 

In this subsection, we examine and discuss our main result \Cref{thm:main} from the perspective of sublinear-time approximation algorithms.

\subsubsection{MaxCSP Problems}

A natural approximation problem associated with CSPs is that of estimating the value of a CSP instance (as defined in \Cref{def:CSP_value}).\footnote{Ideally, an approximation algorithm would also produce an assignment achieving this value, i.e. one that satisfies as many constraints as possible. However, since we focus on sublinear-time algorithms in this paper, we cannot expect them to explicitly output a complete assignment.} This estimation task can be conveniently formulated in the following ``gap version.''

\begin{definition}
For a fixed CSP template $(D,\Gamma)$, a completeness parameter $c\in (0,1]$ and a soundness parameter $s\in [0,c)$, the problem $\Mcsp{\Gamma}{c}{s}$ is the promise problem where given an instance $I\in\CSP(\Gamma)$, 
the (randomized) algorithm should distinguish between the following two cases:
\begin{enumerate}[label = (\arabic*)]
    \item Yes case: if $\val_{I}\geq c$, then 
    the algorithm should accept with probability at least $2/3$.
    \item No case: if $\val_{I}\leq s$, then the algorithm should reject with probability at least $2/3$.   
\end{enumerate}
\end{definition}

The complexity of MaxCSP problems has been extensively studied in the context of polynomial-time algorithms, and a central question is to determine for which tuples $(\Gamma, c, s)$ the problem $\Mcsp{\Gamma}{c}{s}$ can be solved in polynomial time. Notably, Raghavendra \cite{raghavendra2008optimal} almost\footnote{Except for the perfect-completeness case (i.e. the case $c=1$), which remains a mystery.} characterizes all such tuples, conditioned on the Unique Games Conjecture of \cite{khot2002power}.

\subsubsection{MaxCSP vs. Satisfiability Testing}\label{subsubsec:MaxCSP_vs_testing}

Whether property testing of CSP satisfiability has implications for MaxCSP problems depends on the notion of “distance to satisfiability” used by the property tester. In the $d$-bounded-degree model, having distance $\varepsilon$ to satisfiability means that at least $\varepsilon dn$ constraints must be removed to make the instance satisfiable. In contrast, an instance with value $1-\varepsilon$ in the MaxCSP sense requires removing at least an $\varepsilon$ \emph{fraction} of its constraints. Thus, testing CSP satisfiability in the bounded-degree model is analogous to solving $\Mcsp{\Gamma}{1}{1-\varepsilon}$, provided the CSP instance has $\Theta(n)$ constraints.

However, a CSP instance in the $d$-bounded-degree model $\BD(d,n)$ may have far fewer than $\Theta(n)$ constraints. In fact, if an instance contains only $O(1)$ constraints, it would be unreasonable to expect any sublinear-time algorithm to approximate its value, since with high probability the queried variables would not appear in any constraint. To address this, we introduce a modified model $\BD^*(d,\alpha,n)$, where $\alpha > 0$ is a fixed constant ensuring that all CSP instances considered contain at least $\alpha n$ constraints.

Under this refinement, \Cref{thm:main} can be equivalently stated as follows:

\begin{theorem}\label{thm:main_rephrase}
For any CSP template $(D,\Gamma)$ of unbounded width, there exist constants $\varepsilon,\alpha\in (0,1)$ and $d\in \mathbb{N}$ such that any algorithm for $\Mcsp{\Gamma}{1}{1-\varepsilon}$ in the $\BD^*(d,\alpha,n)$ model must make $\Omega(n)$ queries.
\end{theorem}

It is easy to see that \Cref{thm:main_rephrase} implies \Cref{thm:main}. Due to the convenience provided by this rephrased version, in the rest of the paper we will work with \Cref{thm:main_rephrase} instead of \Cref{thm:main}.

\subsubsection{Sublinear Time vs. Polynomial Time}\label{subsubsec:polytime}

As discussed in \Cref{subsubsec:MaxCSP_vs_testing}, because of the specific notion of “distance to satisfiability” used in the bounded-degree model, satisfiability testing in this setting most naturally corresponds to MaxCSP problems on instances with $\Theta(n)$ constraints. In contrast, the \emph{dense graph model} (the more extensively studied model in graph property testing) uses another notion of distance, making it better aligned with MaxCSP problems on ``dense'' instances. 

We argue that exactly due to this difference in instance regimes, the bounded-degree model provides a stronger connection to the study of polynomial-time algorithms for MaxCSPs, at least with respect to worst-case approximation ratios. First, MaxCSP problems on general instances can often be reduced to bounded-degree instances (see, e.g.,~\cite{trevisan2001non}), so bounded-degree instances already capture much of the hardness in the problems. Second, MaxCSP problems tend to be easier on dense instances (see, e.g.,~\cite{arora1995polynomial}). Indeed,~\cite{alon2003testing} shows that every CSP admits a constant-query satisfiability tester in the dense graph model. In contrast, in the bounded-degree model, satisfiability testing already exhibits a wide spectrum of complexities across CSP templates, with query complexity ranging from constant, to $\widetilde{\Theta}(\sqrt{n})$, up to $\Theta(n)$. 

We next discuss some results about MaxCSP in the polynomial-time setting that are most relevant to this work. 

In the context of polynomial-time algorithms, for certain CSP templates \((D,\Gamma)\), the \emph{perfect-completeness} problem \(\Mcsp{\Gamma}{1}{1-\varepsilon}\) can exhibit a markedly different level of complexity from the (potentially harder) \emph{imperfect-completeness} problem \(\Mcsp{\Gamma}{1-\varepsilon'}{1-\varepsilon}\), even when \(\varepsilon' \to 0\) and \(\varepsilon\) is fixed. 
For example, in the case of \(\mathsf{3LIN}\) (an unbounded-width CSP), the former problem admits a polynomial-time algorithm via Gaussian elimination, whereas the latter is \(\mathbf{NP}\)-hard (due to \cite{haastad2001some}). 
More generally, Dalmau and Krokhin~\cite{dalmau2013robust} proved that the imperfect-completeness problem is \(\mathbf{NP}\)-hard for all unbounded-width CSP templates.

\begin{theorem}[\cite{dalmau2013robust}]\label{thm:DK13}
For any CSP template \((D,\Gamma)\) of unbounded width, there exists a constant \(\varepsilon \in (0,1)\) such that for every constant \(\varepsilon' \in (0,\varepsilon)\), the problem \(\Mcsp{\Gamma}{1-\varepsilon'}{1-\varepsilon}\) is \(\mathbf{NP}\)-hard.
\end{theorem}

In terms of conclusions, our result is not directly comparable with \Cref{thm:DK13}, for multiple reasons. First, \Cref{thm:DK13} is a hardness result for the imperfect completeness problem, while our \Cref{thm:main_rephrase} shows hardness for the (potentially easier) perfect completeness problem. Second, \Cref{thm:DK13} is proved against polynomial-time algorithms, while \Cref{thm:main_rephrase} is against sublinear-query algorithms. Although $\mathbf{NP}$-hardness is likely much stronger than a linear query lower bound, the latter does have the benefit of being unconditional. It is also not clear a priori whether some useful algorithms could have sublinear query complexity but super-polynomial time complexity, so $\mathbf{NP}$-hardness usually does not immediately imply query lower bounds in sublinear-time models, even when assuming $\mathbf{P}\neq\mathbf{NP}$.

In terms of techniques, however, the proof of \Cref{thm:main_rephrase} runs in direct parallel to that of \Cref{thm:DK13} in \cite{dalmau2013robust}. While \cite{dalmau2013robust} constructs a reduction from the $\mathbf{NP}$-hardness of approximating $\mathsf{3LIN}$ instances (established by \cite{haastad2001some}), our proof of \Cref{thm:main_rephrase} builds an analogous reduction from the linear-query lower bound for $\mathsf{3LIN}$ instances (proved in \cite{bogdanov2002lower}). Unlike in the $\mathbf{NP}$ setting, where a reduction need only run in polynomial time, our bounded-degree query setting requires addressing several additional subtleties. In particular, certain parts of the proof (for example, \Cref{subsec:expander_gadget}) rely on a careful combination of techniques from both \cite{bogdanov2002lower} and \cite{dalmau2013robust}.

\begin{remark}\label{rmk:tightness_of_DK13}
The result of \cite{dalmau2013robust} (\Cref{thm:DK13}) was later shown to be tight by Barto and Kozik~\cite{barto2016robustly}. Specifically, they proved that for every CSP template $(D,\Gamma)$ of bounded width and every constant $\varepsilon \in (0,1)$, there exists a constant $\varepsilon' \in (0,\varepsilon)$ such that the problem $\Mcsp{\Gamma}{1-\varepsilon'}{1-\varepsilon}$ belongs to $\mathbf{P}$. Together with the result of \cite{dalmau2013robust}, this establishes that bounded width exactly characterizes the class of \emph{robustly solvable} CSPs, as conjectured earlier by Guruswami and Zhou \cite{guruswami2012tight}.
\end{remark}

\subsubsection{Sublinear Time vs. Sublinear Space}\label{subsubsec:streaming}

The past decades have seen a growing body of work on MaxCSP problems in \emph{sublinear-space} algorithmic models (such as~\cite{kapralov2014streaming,kapralov2019optimal,chou2024sketching,chou2022linear,saxena2025streaming,fei2025dichotomy}, to name a few), particularly in the contexts of sketching and streaming algorithms. Recent studies~\cite{fei2025multi,fei2025dichotomy} reveal a phenomenological connection between the \emph{multi-pass streaming} model and the bounded-degree query model: the class of MaxCSP problems that are easy in the bounded-degree model (admitting algorithms with constantly many queries) almost\footnote{Except for the perfect-completeness case, which remains a mystery for both models.} coincides with those that are easy in the multi-pass streaming model (admitting algorithms with constantly many passes and polylogarithmic space). We view this connection as further motivation for studying MaxCSP and other approximation problems in the bounded-degree query model. 

In light of this connection, a natural question arises: does the hardness result of \Cref{thm:main_rephrase} also extend to the multi-pass streaming model?

\begin{question}\label{ques:streaming}
Is it true that for every CSP template $(D,\Gamma)$ of unbounded width, there exists a constant $\varepsilon > 0$ such that $\Mcsp{\Gamma}{1}{1-\varepsilon}$ requires $\Omega(n)$ space in the multi-pass streaming model?
\end{question}

We remark that, as discussed in \cite{fei2025dichotomy}, query lower bounds in sublinear time models are usually easier to establish than space lower bounds, and thus \Cref{thm:main_rephrase} can be seen as a preliminary step towards answering \Cref{ques:streaming}.

\subsection{Open Problems}

Perhaps the most obvious open problem left by this work is whether the result of \Cref{thm:main_rephrase} is tight, analogous to the case of \Cref{thm:DK13} (see \Cref{rmk:tightness_of_DK13}).

\begin{question}
Do sublinear-query algorithms exist for satisfiability testing of all bounded-width CSPs?
\end{question}

In their original work, Bogdanov, Obata and Trevisan~\cite{bogdanov2002lower} asked whether, given a CSP template $(D,\Gamma)$, one can determine the optimal hardness gap~$\varepsilon$ for which $\Mcsp{\Gamma}{1}{1-\varepsilon}$ requires $\Omega(n)$ queries. This question remains wide open (even for ``simple'' templates such as $\mathsf{3COL}$). Here, we pose the following more ambitious version.

\begin{question}\label{ques:ambitious}
Determine the query complexity of $\Mcsp{\Gamma}{c}{s}$ in the bounded-degree model, for all predicate families $\Gamma$ and parameters $c$ and~$s$.
\end{question}

As discussed in \Cref{subsubsec:streaming}, this question may also have strong relevance to sublinear-space computational models. Moreover, resolving \Cref{ques:ambitious} unconditionally (without assuming $\mathbf{P} \neq \mathbf{NP}$ or the Unique Games Conjecture) could possibly shed new light on the corresponding question for polynomial-time algorithms.

\section{Preliminaries}

As mentioned in \Cref{subsubsec:polytime}, the proof of \Cref{thm:main_rephrase} is by constructing a reduction from ``linear equation'' instances. As is the case in \cite{bogdanov2002lower}, the reduction is centered around the fact that the linear equations can be ``simulated'' by some constant-size gadgets built from constraints in the target template $(D,\Gamma)$. The following general definition captures such gadget constructions.

\begin{definition}\label{def:relational_clone}
Given a relational structure $(D,\Gamma)$ and a relation $R:D^{k}\rightarrow\{0,1\}$, we say that $R$ is \emph{generated} by $\Gamma$ if there exists an instance $I=(V,\cC)\in \CSP(\Gamma)$ such that the following holds:
\begin{enumerate}[label=(\arabic*)]
\item The variable set $V$ is $\{v_{1},\dots,v_{k},u_{1},\dots,u_{\ell}\}$ for some integer $\ell\geq 0$. 
\item For any tuple $(x_{1},\dots,x_{k})\in D^{k}$, we have $R(x_{1},\dots,x_{k})=1$ if and only if there exists a satisfying assignment $\tau:V\rightarrow D$ such that $\tau(v_{i})=x_{i}$ for each $i\in [k]$.
\end{enumerate}
The (possibly infinite) collection of all relations generated by $\Gamma$ is denoted by $\clo{\Gamma}$.
\end{definition}

\begin{remark}
In the setting of \Cref{def:relational_clone}, we will say that the instance $I$ ``simulates'' the relation $R$ on the variables $v_{1},\dots,v_{k}$. These variables will sometimes be referred to as ``primary variables'' later in this paper, while the other variables $u_{1},\dots, u_{\ell}$ in $I$ will be called ``secondary variables.''
\end{remark}

Over the past decades, the development of the powerful framework of universal algebra has enabled a systematic study of the power and limitations of gadget constructions for CSPs. In contrast to the case-by-case constructions used in earlier works such as~\cite{bogdanov2002lower,yoshida2010query,aaronson2025property}, we draw on established results in universal algebra to fully exploit the strength of gadget constructions. The universal-algebraic result (\Cref{thm:main_algebra}) used in the proof of \Cref{thm:main_rephrase} is highly nontrivial (though by now perhaps standard) and serves as the main workhorse that allows us to unify the existing $\Omega(n)$ query lower bounds. The remainder of this section is devoted to introducing the necessary background in universal algebra before formally stating the lemma in the next section.

\subsection{Relational Structures}\label{subsec:relational_structures}

We begin by introducing some basic terminology for studying relational structures.

Recall from \Cref{def:CSP_instance} that we do not allow repetition of variables in individual constraints of CSP instances. For example, if $\textsc{E3Sat}$ is the collection of relations $R:\{0,1\}^{3}\rightarrow\{0,1\}$ such that $|R^{-1}(1)|=7$, then $\CSP(\textsc{E3Sat})$ is the collection of all CNF formulas where each clause contains exactly 3 literals. In order to take into account CNF formulas with some clauses containing less than 3 literals (as is standard in the $\mathsf{3SAT}$ problem), we enlarge the relation set $\textsc{E3Sat}$ into $\textsc{3Sat}$ --- the collection of relations $R:\{0,1\}^{k}\rightarrow\{0,1\}$ such that $|R^{-1}(1)|=2^{k}-1$ and $k\in \{1,2,3\}$. Now $\CSP(\textsc{3Sat})$ expresses the standard version of the $\mathsf{3SAT}$ problem. 

Given a general relation set $\Gamma$, in order to allow variable repetition in $\CSP(\Gamma)$ instances, we can perform the following operation on $\Gamma$.

\begin{definition}
For any relation $R:D^{k}\rightarrow \{0,1\}$ on a finite domain $D$, we let $\overline{\{R\}}$ be the collection of relations $R':D^{\ell}\rightarrow \{0,1\}$ such that there exists a surjective map $\pi:[k]\rightarrow [\ell]$ and $R'(x_{1},\dots,x_{\ell})=R(x_{\pi(1)},\dots,x_{\pi(k)})$ for all $x_{1},\dots,x_{\ell}\in D$. For any relational structure $(D,\Gamma)$, we let $\overline{\Gamma}$ be the union $\bigcup_{R\in \Gamma}\overline{\{R\}}$. A relational structure $(D,\Gamma)$ is said to be \emph{repetition-closed}\footnote{This is not standard terminology.} if $\Gamma=\overline{\Gamma}$.
\end{definition}

By a standard reduction, in our main theorem (\Cref{thm:main_rephrase}) it suffices to consider relational structures that are repetition-closed. Formally, it will be proved in \Cref{sec:allow_repeition} that the following statement implies \Cref{thm:main_rephrase}.

\begin{theorem}\label{thm:main_repetition}
For any repetition-closed CSP template $(D,\Gamma)$ of unbounded width, there exist constants $\varepsilon,\alpha\in (0,1)$ and $d\in \mathbb{N}$ such that any algorithm for $\Mcsp{\Gamma}{1}{1-\varepsilon}$ in the $\BD^*(d,\alpha,n)$ model must make $\Omega(n)$ queries.
\end{theorem}

We next turn to study maps between relational structures.

\begin{definition}\label{def:relational_homomorphism}
Let $(D_{1},\Gamma_{1})$ and $(D_{2},\Gamma_{2})$ be relational structures.  
Suppose there exists a bijection between $\Gamma_{1}$ and $\Gamma_{2}$ such that paired relations have the same arity.  
A map $\varphi : D_{1} \to D_{2}$ is called a \emph{relational homomorphism} if for every $R_{1} \in \Gamma_{1}$ and $R_{2} \in \Gamma_{2}$ paired under this bijection with arity $k$, and for all $x_{1},\dots,x_{k}\in D_{1}$,
\[
R_{1}(x_{1},\dots,x_{k}) = 1 
\quad\implies\quad 
R_{2}\big(\varphi(x_{1}),\dots,\varphi(x_{k})\big) = 1.
\]
\end{definition}

Throughout the rest of the paper, we consider maps between relational structures only when the relations of the source and target structures are bijectively matched, and corresponding relations have the same arity. Such pairs of structures are said to share the same \emph{signature}.

\begin{definition}\label{def:core}
Let $(D,\Gamma)$ be a relational structure. A map $\varphi:D\rightarrow D$ is called a \emph{relational endomorphism} if it is a relational homomorphism from $(D,\Gamma)$ to itself (under the obvious identification of relations). A relational endomorphism is called a \emph{relational automorphism} if it is bijective and its inverse is also a relational endomorphism.\footnote{The second condition is actually unnecessary because we only deal with finite structures. In the finite setting, the inverse of a bijective relational endomorphism is automatically an endomorphism.} A relational structure is called a \emph{core} if all of its relational endomorphisms are relational automorphisms. 
\end{definition}

The next task is to show that every relational structure can be ``reduced'' to a core. 

\begin{definition}
Two relational structures $(D_{1},\Gamma_{1})$ and $(D_{2},\Gamma_{2})$ are said to be \emph{homomorphically equivalent} if there exist two relational homomorphisms $\varphi_{1}:D_{1}\rightarrow D_{2}$ and $\varphi_{2}:D_{2}\rightarrow D_{1}$ (under the same identification between $\Gamma_{1}$ and $\Gamma_{2}$). 
\end{definition}

\begin{proposition}\label{prop:core_existence}
Every relational structure $(D,\Gamma)$ is homomorphically equivalent to a core. If $(D,\Gamma)$ is repetition-closed, then $(D,\Gamma)$ is homomorphically equivalent to a repetition-closed core.
\end{proposition}

\begin{proof}
Let $D_{1}\subseteq D$ be a smallest subset for which there exists a relational homomorphism $\varphi$ from $(D,\Gamma)$ to $(D_{1},\Gamma_{1})$, where $\Gamma_{1}$ is obtained by restricting the relations in $\Gamma$ to $D_{1}$. Since the natural embedding $\iota:D_{1}\hookrightarrow D$ is itself a relational homomorphism, we conclude that $(D,\Gamma)$ and $(D_{1},\Gamma_{1})$ are homomorphically equivalent. 

We claim that $(D_{1},\Gamma_{1})$ is a core. Otherwise, there exists a relational endomorphism $\varphi_{1}$ on $(D_{1},\Gamma_{1})$ that is not bijective. In particular, its image is a proper subset $D_{2}\subsetneq D_{1}$. But then the composition $\varphi_{1}\circ\varphi$ defines a homomorphism from $D$ to $D_{2}$, contradicting the minimality of $D_{1}$. Hence $(D_{1},\Gamma_{1})$ must be a core.

Now assume $(D,\Gamma)$ is repetition closed. Any relation in $\Gamma_{1}$ of arity $k$ is the restriction of some relation $R:D^{k}\rightarrow\{0,1\}$ to $D_{1}$. Therefore, for any surjection $\pi:[k]\rightarrow[\ell]$, there exists $R'\in \Gamma$ such that $R'(x_{1},\dots,x_{\ell})=R(x_{\pi(1)},\dots,x_{\pi(k)})$ for all $x_{1},\dots,x_{\ell}\in D$, and in particular $R'_{|D_{1}}(x_{1},\dots,x_{\ell})=R_{|D_{1}}(x_{\pi(1)},\dots,x_{\pi(k)})$ for all $x_{1},\dots,x_{\ell}\in D_{1}$. As $R'_{|D_{1}}\in \Gamma_{1}$, we conclude that $(D_{1},\Gamma_{1})$ is also repetition closed.
\end{proof}

\subsection{Algebraic Structures}

Universal algebra enters the picture through the notion of \emph{algebraic structures}.

\begin{definition}
An \emph{algebraic structure} is a pair $(D,\cF)$ where $D$ is a finite set and $\cF$ is a (possibly infinite) collection of operations on $D$. Each operation $f \in \cF$ is a function $f : D^k \to D$ for some positive integer $k$ (the \emph{arity} of $f$).
\end{definition}

Although at first glance algebraic structures may seem to run parallel to relational structures, the two are in fact deeply interconnected through the following key definitions.

\begin{definition}\label{def:preserve_relations}
Let $f : D^k \to D$ be an operation, and let $R : D^{\ell} \to \{0,1\}$ be a relation on $D$. We say that $f$ \emph{preserves} $R$ if for every matrix $(x_{i,j}) \in D^{\ell \times k}$,
\[
R(x_{1,j}, \dots, x_{\ell,j}) = 1 \text{ for all } j \in [k]
\]
implies
\[
R\big(f(x_{1,1}, \dots, x_{1,k}), \dots, f(x_{\ell,1}, \dots, x_{\ell,k})\big) = 1.
\]
\end{definition}

\begin{definition}\label{def:polymorphisms}
An operation $f$ is a \emph{polymorphism} of a relational structure $(D,\Gamma)$ if it preserves every relation $R \in \Gamma$. The (infinite) set of all polymorphisms of $(D,\Gamma)$ is denoted by $\Pol(\Gamma)$.
\end{definition}

Note that a unary polymorphism is exactly a relational endomorphism, as defined in \Cref{def:relational_homomorphism}.

The following definition gives algebraic structures their own homomorphism maps. It is important not to confuse algebraic homomorphisms with relational homomorphisms.

\begin{definition}\label{def:algebraic_homomorphism}
Let $(D_{1},\cF_{1})$ and $(D_{2},\cF_{2})$ be algebraic structures.  
Suppose there exists a bijection between $\cF_{1}$ and $\cF_{2}$ such that paired functions have the same arity.  
A map $\varphi : D_{1} \to D_{2}$ is called an \emph{algebraic homomorphism} if for every $f_{1}\in \cF_{1}$ and $f_{2}\in \cF_{2}$ paired under this bijection with arity $k$, and for all $x_{1},\dots,x_{k}\in D_{1}$,
\[
\varphi\big(f_{1}(x_{1},\dots,x_{k})\big)= f_{2}\big(\varphi(x_{1}),\dots,\varphi(x_{k})\big).
\]
\end{definition}

\begin{definition}\label{def:subalgebra}
Let $(D_{1},\cF_{1})$ and $(D_{2},\cF_{2})$ be algebraic structures.  
\begin{enumerate}[label=(\arabic*)]
\item We say that $(D_{1},\cF_{1})$ is a subalgebra of $(D_{2},\cF_{2})$ if there exists an injective algebraic homomorphism $\varphi:D_{1}\rightarrow D_{2}$. 
\item We say that $(D_{2},\cF_{2})$ is a homomorphic image of $(D_{1},\cF_{1})$ if there exists a surjective algebraic homomorphism $\varphi:D_{1}\rightarrow D_{2}$.
\end{enumerate}
\end{definition}

It turns out that a special class of algebraic structures, namely the \emph{idempotent} ones, provides significant convenience in various contexts. Since idempotent structures play an important role in formulating the main universal-algebraic lemma (\Cref{lem:main_algebra}), we give the definition below.

\begin{definition}
An operation $f:D^{k}\rightarrow D$ is \emph{idempotent} if $f(x,x,\dots,x)=x$ holds for all $x\in D$. We call algebraic structure $(D,\cF)$ idempotent if every operation $f\in \cF$ is idempotent.
\end{definition}

\subsection{Irredundant Relations and Galois Duality}

Given a relational structure $(D, \Gamma)$, let $\langle \Gamma \rangle$ denote the (possibly infinite) set of relations that are preserved by $\Pol(\Gamma)$. It is not hard to see that every relation generated by $\Gamma$ in the sense of \Cref{def:relational_clone} is preserved by $\Pol(\Gamma)$.\footnote{This observation is not used elsewhere in this paper.} Using the notation of \Cref{def:relational_clone}, this implies $\clo{\Gamma} \subseteq \langle \Gamma \rangle$.

A fundamental result of Geiger~\cite{geiger1968closed} shows that if the notion of “relations generated by~$\Gamma$” in \Cref{def:relational_clone} is suitably relaxed, then the resulting set $\clo{\Gamma}$ actually coincides with $\langle \Gamma \rangle$. This result captures a deep “duality” between algebraic and relational structures (an instance of the more general concept of a \emph{Galois connection}), which we do not elaborate on here. Instead, following Dalmau and Krokhin~\cite{dalmau2013robust}, we state a variant of Geiger’s theorem asserting that any \emph{irredundant} relation in $\langle \Gamma \rangle$ is contained in $\clo{\Gamma}$ (without relaxing \Cref{def:relational_clone}).

\begin{definition}\label{def:irredundance}
Let $R:D^{k}\rightarrow\{0,1\}$ be a relation on a finite set $D$. We say that $R$ is \emph{irredundant} if for any distinct indices $i,j\in [k]$, there exists a tuple $(x_{1},\dots,x_{k})\in D^{k}$ such that $R(x_{1},\dots,x_{k})=1$ and $x_{i}\neq x_{j}$.
\end{definition}

\begin{theorem}[\cite{geiger1968closed}]\label{thm:Galois}
Let $(D,\Gamma)$ be a repetition-closed relational structure, and let $R:D^{k}\rightarrow\{0,1\}$ be an irredundant relation on $D$. If every polymorphism of $(D,\Gamma)$ preserves $R$, then $R\in \clo{\Gamma}$.
\end{theorem}

For the sake of completeness, a proof of \Cref{thm:Galois} is given in \Cref{sec:Galois_proof}. 

As an application of \Cref{thm:Galois}, we next define an ``Endomorphism'' relation for any given relational structure $(D,\Gamma)$ and prove that it belongs to $\clo{\Gamma}$.

\begin{definition}\label{def:End_relation}
For any relational structure $(D,\Gamma)$, we define a relation $End_{\Gamma}:D^{D}\rightarrow \{0,1\}$ as follows. For any tuple $y=(y_{x})_{x\in D}\in D^{D}$, we let $End_{\Gamma}(y)=1$ if and only if the map $x\mapsto y_{x}$ is a relational endomorphism on $(D,\Gamma)$.
\end{definition}

\begin{proposition}\label{prop:End_in_clo}
For any repetition-closed relational structure $(D,\Gamma)$, we have $End_{\Gamma}\in \clo{\Gamma}$. 
\end{proposition}

\begin{proof}
It is easy to see that $End_{\Gamma}$ is irredundant, since the tuple $y\in D^{D}$ defined by $y_{x}=x$ for all $x\in D$ satisfies the relation $End_{\Gamma}$. It suffices to show that $End_{\Gamma}$ is preserved by all polymorphisms of $\Gamma$. Once we have that, \Cref{thm:Galois} implies that $End_{\Gamma}\in\clo{\Gamma}$.

Suppose $f\in \Pol(\Gamma)$ is a $k$-ary polymorphism. For any $k$ tuples $y^{(1)},\dots,y^{(k)}\in D^{D}$ satisfying $End_{\Gamma}$, there are endomorphisms $\varphi_{1},\dots,\varphi_{k}$ of $(D,\Gamma)$ such that $\varphi_{i}(x)=y^{(i)}_{x}$ for all $i\in [k]$ and $x\in D$. It now suffices to verify that the tuple
\[
\Big(f\big(\varphi_{1}(x),\dots,\varphi_{k}(x)\big)\Big)_{x\in D}\in D^{D}
\]
satisfies $End_{\Gamma}$, or in other words, the map $x\mapsto f\big(\varphi_{1}(x),\dots,\varphi_{k}(x)\big)$ is a relational endomorphism of $(D,\Gamma)$. But this is a special case of the general observation that for any relational structure, the collection of its polymorphisms is closed under composition.\footnote{In general, a $k$-ary polymorphism $f$ composed with an $m$-ary polymorphism $g$ yields a $km$-ary polymorphism $f\circ g$ defined by $f\circ g(x_{1,1},\dots,x_{1,m},\dots,x_{k,1},\dots,x_{k,m})=f(g(x_{1,1},\dots,x_{1,m}),\dots,g(x_{k,1},\dots,x_{k,m}))$} 
\end{proof}

\section{Main Lemma from Universal Algebra}\label{sec:universal_algebra_lemma}

In this section, we state the main universal-algebraic result used in the proof of \Cref{thm:main_repetition}. Informally, the result asserts that any CSP template of unbounded width can simulate relations defined by linear equations over abelian groups, in the sense of \Cref{def:relational_clone}. The formal statement, however, is much more nuanced.

We begin by defining $\TLin_{\mathsf{G}}$, the set of relations that specify the sum of exactly three variables. Throughout the rest of the paper, Abelian groups denoted by $\mathsf{G}$ are assumed to be finite and have at least two elements.

\begin{definition}
Let $\mathsf{G}$ be an Abelian group. For any element $b\in \mathsf{G}$, we let $S^{\mathsf{G}}_{b}:\mathsf{G}^{3}\rightarrow\{0,1\}$ be the relation defined by
\[S^{\mathsf{G}}_{b}(x_{1},x_{2},x_{3})=1\quad\text{if and only if}\quad
x_{1}+x_{2}+x_{3}=b.
\]
We use \(\TLin_{
\mathsf{G}}\) to denote the relation set $\left\{S_{b}^{\mathsf{G}}\,\middle|\,b\in \mathsf{G}\right\}$.
\end{definition}

We are now ready to state the following main result from universal algebra.

\begin{theorem}[{\cite{barto2014constraint}}]\label{thm:main_algebra}
Let \((D,\Gamma)\) be a relational structure of unbounded width.  
Assume further that \((D,\Gamma)\) is a core and that all operations in \(\Pol(\Gamma)\) are idempotent.  
Then there exists an Abelian group $\mathsf{G}$ and a family of operations \(\mathcal{F}\subseteq \Pol(\TLin_{\mathsf{G}})\) such that $(\mathsf{G},\cF)$ is a homomorphic image of a subalgebra of the polymorphism algebra \((D,\Pol(\Gamma))\) .
\end{theorem}

The statement of \Cref{thm:main_algebra} is primarily derived from~\cite[Conjecture~4.3]{barto2014constraint}, with additional ingredients drawn from other works, such as~\cite{valeriote2009subalgebra}. Further details concerning \Cref{thm:main_algebra} will be provided in \Cref{sec:varieties}.

The remaining goal of this section is to translate \Cref{thm:main_algebra} into the more explicit and directly applicable statement of \Cref{lem:main_algebra}. One aspect of \Cref{thm:main_algebra} that requires ``translation'' is the idempotency condition on $\Pol(\Gamma)$. Given a relational structure, its polymorphism algebra can be made idempotent by augmenting the relation set with all constant relations, defined as follows.

\begin{definition}
Let \(D\) be a finite set.  
We denote by \(\Const_{D}\) the set of all unary constant relations on \(D\); that is, relations \(R\) of the form 
\[
R(x)=1\quad\text{if and only if}\quad x=b,
\]
for some \(b\in D\).
\end{definition}

We are now ready to state and prove the following translated version of \Cref{thm:main_algebra}.

\begin{lemma}\label{lem:main_algebra}
Let \((D,\Gamma)\) be a repetition-closed relational structure of unbounded width.  
Then there exist a subset \(D'\subseteq D\), an Abelian group $\mathsf{G}$, and a surjective map \(\varphi\colon D'\to\mathsf{G}\) such that any ternary relation \(R\colon D^{3}\to\{0,1\}\) of the form
\begin{equation}\label{eq:ternary_phi_linear}
R(x_{1},x_{2},x_{3})=1\quad\text{if and only if}\quad
x_{1},x_{2},x_{3}\in D'\text{ and }
R'(\varphi(x_{1}),\varphi(x_{2}),\varphi(x_{3}))=1,
\end{equation}
for some $R'\in \TLin_{\mathsf{G}}$,
belongs to the collection \(\clo{\Gamma\cup\Const_{D}}\).
\end{lemma}

\begin{proof}

We divide the proof into the following steps.

\paragraph{Step 1: construction of $\varphi$.} Because of the presence of constant relations, every relational endomorphism of the structure \((D,\Gamma\cup\Const_{D})\) is the identity map.  
By \Cref{def:core}, this implies that \((D,\Gamma\cup\Const_{D})\) is a core.  
Furthermore, by \Cref{prop:Const_preserves_unbounded_width}, \((D,\Gamma\cup\Const_{D})\) also has unbounded width.  
Since every polymorphism of \((D,\Gamma\cup\Const_{D})\) preserves the constant relations and is therefore idempotent, we can apply \Cref{thm:main_algebra} to obtain a subalgebra \((D_{1},\mathcal{F}_{1})\) of \(\bigl(D,\Pol(\Gamma\cup\Const_{D})\bigr)\) and a homomorphic image \((\mathsf{G},\mathcal{F}_{2})\) of \((D_{1},\mathcal{F}_{1})\) such that \(\mathcal{F}_{2}\subseteq\Pol(\TLin_{\mathsf{G}})\), as illustrated in the following diagram.
\[
\begin{tikzcd}[row sep=0em, column sep=6em]
  & \bigl(D,\Pol(\Gamma\cup\Const_D)\bigr) \\
  (D_{1},\cF_{1}) \arrow[ur, hook, "\varphi_{1}"] 
                \arrow[dr, two heads, "\varphi_{2}"] & \\
  & (\mathsf{G},\cF_2)
\end{tikzcd}
\]

Let $\varphi_{1}:D_{1}\rightarrow D$ be the injective algebraic homomorphism from $(D_{1},\cF_{1})$ to $\big(D,\Pol(\Gamma\cup\Const_{D})\big)$, and let $\varphi_{2}:D_{1}\rightarrow \mathsf{G}$ be the surjective algebraic homomorphism from $(D_{1},\cF_{1})$ to $(\mathsf{G},\cF_2)$. Let $D'=\varphi_{1}(D_{1})\subseteq D$ be the image of $\varphi_{1}$. Then we may define a surjective map $\varphi:D'\rightarrow \mathsf{G}$ by letting $\varphi(x)=\varphi_{2}(\varphi_{1}^{-1}(x))$ for any $x\in D'$. 

\paragraph{Step 2: properties of $\varphi$.} We claim that any polymorphism $f:D^{k}\rightarrow D$ of the relational structure $(D,\Gamma\cup\Const_{D})$ also preserves any relation $R$ of the form \eqref{eq:ternary_phi_linear}. Since $\varphi_{1},\varphi_{2}$ are algebraic homomorphisms, by \Cref{def:algebraic_homomorphism} we know that some $f_{1}\in \cF_{1}$ is implicitly paired with $f$ and some $f_{2}\in \cF_2$ is in turn paired with $f_{1}$. 

Let $(x_{i,j})\in D^{3\times k}$ be a matrix of elements such that $R(x_{1,j},x_{2,j},x_{3,j})=1$ for all $j\in [k]$. Therefore, by \eqref{eq:ternary_phi_linear} it must be the case that $x_{\ell,j}\in D'$ for all $\ell\in [3]$ and $j\in [k]$. We let $y_{\ell,j}=\varphi_{1}^{-1}(x_{\ell,j})\in D_{1}$. Now we have
\begin{align}
&\quad R\bigl(f(x_{1,1},\dots,x_{1,k}),f(x_{2,1},\dots,x_{2,k}), f(x_{3,1},\dots,x_{3,k})\bigr)\nonumber\\
&= R\Bigl(f\big(\varphi_{1}(y_{1,1}),\dots,\varphi_{1}(y_{1,k})\big),f\big(\varphi_{1}(y_{2,1}),\dots,\varphi_{1}(y_{2,k})\big), f\big(\varphi_{1}(y_{3,1}),\dots,\varphi_{1}(y_{3,k})\big)\Bigr)\nonumber\\
&= R'\Big(\varphi_{2}\circ\varphi_{1}^{-1}\Big(f\big(\varphi_{1}(y_{1,1}),\dots,\varphi_{1}(y_{1,k})\big)\Big),\varphi_{2}\circ\varphi_{1}^{-1}\Big(f\big(\varphi_{1}(y_{2,1}),\dots,\varphi_{1}(y_{2,k})\big)\Big),\nonumber\\
&\quad\quad\quad\varphi_{2}\circ\varphi_{1}^{-1}\Big(f\big(\varphi_{1}(y_{3,1}),\dots,\varphi_{1}(y_{3,k})\big)\Big)\Big)\nonumber\\
&=R'\Big(\varphi_{2}\big(f_{1}(y_{1,1},\dots,y_{1,k})\big),\varphi_{2}\big(f_{1}(y_{2,1},\dots,y_{2,k})\big),\varphi_{2}\big(f_{1}(y_{3,1},\dots,y_{3,k})\big)\Big)\nonumber\\
&=R'\Big(f_{2}\big(\varphi_{2}(y_{1,1}),\dots,\varphi_{2}(y_{1,k})\big),f_{2}\big(\varphi_{2}(y_{2,1}),\dots,\varphi_{2}(y_{2,k})\big),f_{2}\big(\varphi_{2}(y_{3,1}),\dots,\varphi_{2}(y_{3,k})\big)\Big)\nonumber\\
&=R'\Big(f_{2}\big(\varphi(x_{1,1}),\dots,\varphi(x_{1,k})\big),f_{2}\big(\varphi(x_{2,1}),\dots,\varphi(x_{2,k})\big),f_{2}\big(\varphi(x_{3,1}),\dots,\varphi(x_{3,k})\big)\Big)\label{eq:evaluates_to_1}
\end{align}
Since $\cF_{2}\subseteq \Pol(\TLin_{\mathsf{G}})$ and $f_{2}\in \cF_{2}$, it follows from the assumption $R'\in \TLin_{\mathsf{G}}$ that $f_{2}$ preserves $R'$. Combining this with the observation that \[R'\big(\varphi(x_{1,j}),\varphi(x_{2,j}),\varphi(x_{3,j})\big)=R(x_{1,j},x_{2,j},x_{3,j})=1\quad\text{for all }j\in [k],\]
we deduce that \eqref{eq:evaluates_to_1} evaluates to 1. Therefore, we conclude that $f$ preserves $R$.

\paragraph{Step 3: applying Galois duality.} Since every relation $R'\in \TLin_{\mathsf{G}}$ is irredundant in the sense of \Cref{def:irredundance}, it is easy to see that any relation $R:D^{3}\rightarrow\{0,1\}$ defined according to \eqref{eq:ternary_phi_linear} is also irredundant. Due to \Cref{thm:Galois} and the claim proved in Step 2 of the proof, it follows that $R\in \clo{\Gamma\cup\Const_{D}}$.
\end{proof}

\section{Hardness of Linear Equations}\label{sec:BOT}

As mentioned earlier, we prove \Cref{thm:main_repetition} by constructing a reduction from the following hardness result for the ``linear equation'' template $(\mathsf{G},\TLin_{\mathsf{G}})$.

\begin{theorem}[\cite{bogdanov2002lower}]\label{thm:BOT}
For any Abelian group $\mathsf{G}$ and any $\varepsilon>0$, there exists a constant $d>0$ such that $\Mcsp{\TLin_{\mathsf{G}}}{1}{\frac{1}{|\mathsf{G}|}+\varepsilon}$ requires $\Omega(n)$ queries in the $\BD^*(d,\frac{d}{3},3n)$ model.\footnote{The paper \cite{bogdanov2002lower} only proved \Cref{thm:BOT} in the case $\mathsf{G}=\mathbb{Z}/2\bZ$, but the proof can easily be adapted to work for general Abelian groups.}
\end{theorem}

However, unlike \cite{dalmau2013robust}, which applies the result of \cite{haastad2001some} as a black box to prove \Cref{thm:DK13}, we cannot directly use \Cref{thm:BOT} due to the subtleties of our bounded-degree query model. In this section, we unpack the proof of \Cref{thm:BOT} and extract the components that can be explicitly leveraged in proving \Cref{thm:main_repetition}. For the sake of convenience, several aspects of the proof in \cite{bogdanov2002lower} are slightly modified in the presentation below.

The proof of \Cref{thm:BOT} proceeds by defining two probability distributions over $\CSP(\TLin_{\mathsf{G}})$ instances, referred to as the YES and NO distributions. Instances drawn from the YES distribution typically have value 1, while those from the NO distribution typically have value at most $\frac{1}{|\mathsf{G}|}+\varepsilon$. The goal is to show that any sublinear-query algorithm cannot reliably distinguish between instances sampled from these two distributions.

A combinatorial structure used in both the YES and NO distributions is hypergraph perfect matching, formally defined as follows.

\begin{definition}\label{def:perfect_matching}
Let \(n\) be a positive integer, and let $D$ be a finite set of cardinality $k\geq 2$. Consider the complete \(k\)-partite \(k\)-uniform hypergraph with vertex set \([n]\times D\), partitioned according to the second coordinate.  
We write \(\mathcal{M}_{D,n}\) for the collection of all perfect matchings of this hypergraph (a \emph{perfect matching} is a set of \(n\) hyperedges, no two of which share a common vertex). In other words, a matching $M\in \cM_{D,n}$ can be specified by $k$ permutations $\pi_{x}:[n]\rightarrow [n]$, for $x\in D$, where $M$ is the edge set
\[
M=\left\{\big((\pi_{x}(i),x)\big)_{x\in D}\in \prod_{x\in D}\big([n]\times \{x\}\big)\,\middle|\,i\in [n]\right\}
\]
\end{definition}

A perfect matching $M \in \cM_{D,n}$ serves as a guide for placing constraints on variables. Viewing the vertex set $[n] \times D$ as a collection of variables, each hyperedge of $M$ corresponds to a tuple of variables on which a constraint can be imposed. In the setting of \Cref{thm:BOT}, we take $D = [3]$, and on each hyperedge $(v_{1}, v_{2}, v_{3})$ we place a constraint from $\TLin_{\mathsf{G}}$.

The following definition specifies the YES and NO distributions used to prove \Cref{thm:BOT}.

\begin{definition}\label{def:Dyes}
For any positive integers $n,d$ and any Abelian group $\mathsf{G}$, we define $\Dpair^{\mathsf{G},n,d}$ to be the distribution over pairs of $\CSP\big(\TLin_{\mathsf{G}},[n]\times [3]\big)$ instances obtained by sampling from the following procedure: 
\begin{enumerate}
\item Sample an assignment $\tau:[n]\times [3]\rightarrow\mathsf{G}$ uniformly at random.
\item Initialize empty instances $I_{\yes},I_{\no}$. Repeat $d$ times:
\begin{itemize}
\item Sample uniformly at random a $3$-partite perfect matching $M\in \cM_{[3],n}$.
\item For each hyperedge $(v_{1},v_{2},v_{3})\in M$:
\begin{itemize}
\item Add to $I_{\yes}$ the constraint $\big((v_{1},v_{2},v_{3}),S^{\mathsf{G}}_{b}\big)$, where \(b=\tau(v_{1})+\tau(v_{2})+\tau(v_{3})\).
\item Sample $b'$ uniformly from $\mathsf{G}$, and add to $I_{\no}$ the constraint $\big((v_{1},v_{2},v_{3}),S^{\mathsf{G}}_{b'}\big)$.
\end{itemize}
\end{itemize}
\item Output the pair $(I_{\yes},I_{\no})$.
\end{enumerate}
The marginal distributions of $\Dpair^{\mathsf{G},n,d}$ on the first and second components are denoted by $\Dyes^{\mathsf{G},n,d}$ and $\Dno^{\mathsf{G},n,d}$, respectively.
\end{definition}

By construction, the distributions $\Dyes^{\mathsf{G},n,d}$ and $\Dno^{\mathsf{G},n,d}$ satisfy the properties stated in the following proposition.

\begin{proposition}\label{prop:Dpair}
In every instance \(I\) in the support of \(\mathcal{D}_{\mathrm{yes}}^{\mathsf{G},n,d}\) or \(\mathcal{D}_{\mathrm{no}}^{\mathsf{G},n,d}\), each variable \(v\in[n]\times [3]\) is involved in exactly $d$ constraints. Furthermore, every instance \(I\) in the support of \(\mathcal{D}_{\mathrm{yes}}^{\mathsf{G},n,d}\) is satisfiable.
\end{proposition}

The next lemma shows that instances sampled from the NO distribution typically have value at most $\frac{1}{|\mathsf{G}|}+\varepsilon$, as expected.

\begin{lemma}[\cite{bogdanov2002lower}]\label{lem:3LIN_exists_d}
For any Abelian group $\mathsf{G}$ and any \(\varepsilon>0\), there exist a constant integer \(d\geq 1\) such that 
\[
\Pru{I\sim \mathcal{D}_{\mathrm{no}}^{\mathsf{G},n,d}}{\val_{I}\geq \frac{1}{|\mathsf{G}|}+\varepsilon}=o(1),
\]
where \(o(1)\) denotes a term tending to \(0\) as \(n\to\infty\).
\end{lemma}

It remains to show that $\Dyes^{\mathsf{G},n,d}$ and $\Dno^{\mathsf{G},n,d}$ are indistinguishable to sublinear-query algorithms. The idea is that, although a random instance from the YES distribution is globally very different from a random instance from the NO distribution, they are locally very similar. Specifically, when one inspects only a sublinear-sized portion of the instances, the two distributions are actually almost indistinguishable. The following definition helps formalize this intuition.

\begin{definition}
Let \((D,\Gamma)\) be a CSP template and let \(V\) be a set of variables.  
We write \(\CSP(\Gamma,V)\subseteq \CSP(\Gamma)\) for the collection of all \(\CSP(\Gamma)\)-instances whose variable set is exactly \(V\).  
For any subset \(U\subseteq V\), we define the \emph{sub-instance map}  
\[
\Sub_{U}\colon \CSP(\Gamma,V)\longrightarrow \CSP(\Gamma,U)
\]
to be the operation that retains only those constraints whose scopes are contained in \(U\).
\end{definition}

The indistinguishability between $\Dyes^{\mathsf{G},n,d}$ and $\Dno^{\mathsf{G},n,d}$ is formalized in the next lemma.

\begin{lemma}[\cite{bogdanov2002lower}]\label{lem:3LIN_indistinguishable}
For any Abelian group \(\mathsf{G}\) and any positive integer $d$, there exist a constant \(\delta\in (0,1)\) such that the following holds. For any positive integer $n$, there exists an event $\cE$ with $\Dpair^{\mathsf{G},n,d}(\cE)\geq 1-o(1)$ such that for any $U\subseteq [n]\times [3]$ with $|U|\leq \delta n$ and any subset $\Lambda\subseteq \CSP\big(\TLin_{\mathsf{G}},[n]\times [3]\big)$, we have
\[
\Pru{(I_{\yes},I_{\no})\sim \Dpair^{\mathsf{G},n,d}}{\Sub_{U}[I_{\yes}]\in \Lambda \,\Big|\, \cE}
=\Pru{(I_{\yes},I_{\no})\sim \Dpair^{\mathsf{G},n,d}}{\Sub_{U}[I_{\no}]\in \Lambda \,\Big|\, \cE}
\]
\end{lemma}

It is not hard to deduce \Cref{thm:BOT} from \Cref{prop:Dpair} and \Cref{lem:3LIN_exists_d,lem:3LIN_indistinguishable}. For the sake of completeness, proofs of \Cref{lem:3LIN_exists_d,lem:3LIN_indistinguishable} will be sketched in \Cref{sec:random_regular_hypergraph}.

\section{The Reduction Map}

We now arrive at the main technical part of this paper. In this section, we construct the YES and NO distributions for \Cref{thm:main_repetition} using the distributions $\Dyes$ and $\Dno$ introduced in \Cref{sec:BOT}. The construction for the case where $(D, \Gamma)$ is a core is presented in \Cref{subsec:expander_gadget,subsec:construction_in_core} and analyzed in \Cref{subsec:analysis}. Finally, in \Cref{subsec:general_conclusion}, we complete the proof of \Cref{thm:main_repetition} by extending the argument to the general case.

\subsection{The Expander Gadget}\label{subsec:expander_gadget}

Recall that our universal-algebraic lemma (\Cref{lem:main_algebra}) guarantees the simulation of ``linear equation constraints'' only by $\Gamma \cup \Const_{D}$, rather than by $\Gamma$ alone. Consequently, we must first find a way to simulate the constant relations using only the relations in $\Gamma$. Note that constant relations may not be generated by $\Gamma$ in the sense of \Cref{def:relational_clone}, so different techniques are required. Interestingly, the original work of \cite{bogdanov2002lower} encountered the same difficulty when proving the hardness of testing graph 3-colorability.\footnote{We now view this need for constant relations in \cite{bogdanov2002lower} as a special case of the general phenomenon described in \Cref{lem:main_algebra}, since 3-colorability is an example of an unbounded-width CSP.} To address this issue, \cite{bogdanov2002lower} uses an ``expander gadget'' to fix the assignments of certain variables --- a technique that also plays a crucial role in our proof.

In contrast to \cite{bogdanov2002lower}, where ordinary expander graphs suffice, our more general setting requires the use of \emph{random regular hypergraphs}, defined formally below.

\begin{definition}\label{def:random_regular_hypergraph}
Let \(n,\ell\) be positive integers, and let $D$ be a finite set of cardinality $k\geq 2$. Let \(\mathcal{H}_{D,n,\ell}\) denote the distribution of the random \(\ell\)-regular \(k\)-partite hypergraph\footnote{Multi-edges are allowed.} on the vertex set \([n]\times D\), whose edge multiset is the union of \(\ell\) independent perfect matchings chosen uniformly at random from \(\mathcal{M}_{D,n}\).
\end{definition}

Similarly to \Cref{def:Dyes}, the hyperedges of the perfect matchings sampled in \Cref{def:random_regular_hypergraph} are exactly the variable tuples on which we impose constraints.

\begin{definition}\label{def:I_of_RH}
Let $D$ be a finite set of cardinality $k\geq 2$, and let $R:D^{D}\rightarrow\{0,1\}$ be a relation. Any $d$-regular $k$-partite hypergraph $H$ in the support of $\cH_{D,n,\ell}$ uniquely defines a $\CSP(\{R\},[n]\times D)$-instance $\cI[R,H]$, by viewing each hyperedge $\bfv=(v_{x})_{x\in D}$ as a constraint $(\bfv,R)$.
\end{definition}

Ideally, we would like to use a relation $R : D^{D} \to \{0,1\}$ that evaluates to $1$ on a tuple $y = (y_{x})_{x \in D} \in D^{D}$ only when $y_{x} = x$ for all $x \in D$. In this case, the instance $\cI[R, H]$ would admit a unique satisfying assignment that assigns to each variable $(i, x) \in [n] \times D$ its second coordinate, yielding a gadget that effectively simulates constant relations across many variables. However, for many templates $(D, \Gamma)$, such a strict relation cannot be generated in the sense of \Cref{def:relational_clone}. Instead, we settle for weaker relations: it suffices to use a relation $R : D^{D} \to \{0,1\}$ whose satisfying assignments are permutations of $D$. Such relations are formally defined below.

\begin{definition}
Let $D$ be a finite set. A relation $R:D^{D}\rightarrow \{0,1\}$ is called \emph{sub-unique} if any tuple $y=(y_{x})_{x\in D}\in D^{D}$ such that $R(y)=1$ has pairwise distinct components, i.e. $y_{x}\neq y_{x'}$ for any distinct elements $x,x'\in D$.
\end{definition}

Intuitively speaking, when $R:D^{D}\rightarrow\{0,1\}$ is sub-unique and the hypergraph $H$ has a good ``expansion'' property, any assignment that nearly satisfies $\cI[R,H]$ must assign to almost all variables $(i,x)\in [n]\times D$ the value $\pi(x)$, for some permutation $\pi:D\rightarrow D$. The next lemma asserts that a random regular hypergraph $H$ has this desired ``expansion'' property with high probability. Therefore, the instance $\cI[R,H]$ can serve as a gadget that effectively simulates constant relations across many variables, up to only a global permutation.

\begin{lemma}\label{lem:expander_gadget}
Let $D$ be a finite set with at least $2$ elements, and let $R:D^D\rightarrow \{0,1\}$ be a sub-unique relation. For any $\varepsilon\in \big(0,\frac{1}{10|D|}\big]$, there exists a constant integer $\ell\geq 1$ such that with probability $1-\exp(-\Omega(n))$, a random hypergraph $H\sim \cH_{D,n,\ell}$ satisfies the following. For any assignment $\tau:[n]\times D\rightarrow D$ to the instance $\cI[R,H]$ with $\val_{\cI[R,H]}(\tau)\geq 1-\varepsilon$, there exists a tuple $y\in D^{D}$ such that $R(y)=1$ and
\begin{equation}\label{eq:close_to_permutation}
\Pru{(i,x)\in [n]\times D}{\tau\big((i,x)\big)\neq y_{x}}\leq 2\varepsilon.
\end{equation}
\end{lemma}

A proof of \Cref{lem:expander_gadget} is given in \Cref{sec:random_regular_hypergraph}.

\subsection{Construction in the Core Case}\label{subsec:construction_in_core}

Recall that the expander gadget constructed in \Cref{subsec:expander_gadget} relies on the availability of a sub-unique relation. However, some CSP templates cannot generate any sub-unique relation (in the sense of \Cref{def:relational_clone}). Fortunately, when the template $(D, \Gamma)$ is a core (see \Cref{def:core}), the ``Endomorphism'' relation provides a natural sub-unique relation that can be generated from $\Gamma$, as established in \Cref{prop:End_in_clo}.\footnote{The idea of using the Endomorphism relation is due to \cite{dalmau2013robust}.}

In this subsection, we construct the YES and NO distributions for \Cref{thm:main_repetition} in the case the template is a core. Throughout this subsection, we fix a repetition-closed relational structure $(D,\Gamma)$, and assume that it is a core and has unbounded width. In particular, this implies $|D|\geq 2$. 

\subsubsection{Preparatory Materials}\label{subsubsec:prepatory}

We begin by gathering several preparatory materials that we are going to use in the construction. A checklist of these materials is presented below.

\begin{enumerate}[label=(\arabic*)]
\item By applying \Cref{lem:main_algebra}, we obtain a subset \(D'\subseteq D\), an Abelian group \(\mathsf{G}\), and a surjective map \(\varphi:D'\rightarrow \mathsf{G}\). For any $b\in \mathsf{G}$, we let $R_{b}:D^{3}\rightarrow\{0,1\}$ be the relation defined by
\[R_{b}(x_{1},x_{2},x_{3})=1\quad\text{if and only if}\quad
x_{1},x_{2},x_{3}\in D'\text{ and }
S^{\mathsf{G}}_{b}\big(\varphi(x_{1}),\varphi(x_{2}),\varphi(x_{3})\big)=1.\]

\item The conclusion of \Cref{lem:main_algebra} ensures that \(R_{b}\in\clo{\Gamma\cup\Const_{D}}\) for all $b\in \mathsf{G}$. Let $r_{1}$ be a positive integer such that for each $b\in \mathsf{G}$, the relation \(R_{b}\) is simulated by an instance \(\cI_{b}\in \CSP(\Gamma\cup \Const_{D})\) with at most \(r_{1}\) secondary variables and at most \(r_{1}\) constraints (see \Cref{def:relational_clone}). Let $\cI_{b}^{\natural}\in \CSP(\Gamma)$ be the instance obtained by removing all constant constraints from $\cI_{b}$. 

\item We then apply \Cref{lem:3LIN_exists_d} and obtain a positive integer $d$ such that a random $\CSP(\TLin_{\mathsf{G}})$-instance from $\Dno^{\mathsf{G},n,d}$ has value at most $2/3$ with high probability. 

\item By \Cref{prop:End_in_clo}, for some positive integer $r_{2}$, the relation $\End_{\Gamma}$ is simulated by an instance $\cI_{\End}\in \CSP(\Gamma)$ with at most $r_{2}$ secondary variables and at most $r_{2}$ constraints, in the sense of \Cref{def:relational_clone}. 

\item Apply \Cref{lem:expander_gadget} with $\varepsilon=1/(10 r_{1}|D|)$ and $R=\End_{\Gamma}$,\footnote{The relation $\End_{\Gamma}:D^{D}\rightarrow\{0,1\}$ is sub-unique because $(D,\Gamma)$ is a core (see \Cref{def:core}).} which yields a positive integer $\ell$ and a hypergraph $H\in\supp\left(\cH_{D,r_{1}n,\ell}\right)$ satisfying the property stated in \Cref{lem:expander_gadget}.

\end{enumerate}

\subsubsection{The Construction Procedure}\label{subsubsec:construction}

Given a \(\CSP\big(\TLin_{\mathsf{G}}, [n]\times [3]\big)\)-instance \(I\in \supp\big(\mathcal{D}_{\text{yes}}^{\mathsf{G},n,d}\big)\cup\supp\big(\mathcal{D}_{\text{no}}^{\mathsf{G},n,d}\big)\), we transform it into an instance \(\mathfrak{T}[I]\in \CSP(\Gamma)\) as follows.

\paragraph{Step 1: variable sets.}
The variable set of $\mathfrak{T}[I]$ is partitioned into 4 parts:
\begin{itemize}
  \item The set of \emph{original variables} \(V_{\mathrm{orig}}=[n]\times [3]\).
  \item The set of \emph{constant variables} \(V_{\mathrm{const}}=[n]\times [r_{1}]\times D\).
  \item The first set of \emph{auxiliary variables} \(V_{\mathrm{aux}}^{(1)}=[d]\times[n]\times[r_{1}]\).
  \item The second set of \emph{auxiliary variables} \(
  V_{\aux}^{(2)}=[\ell]\times [n]\times [r_{1}]\times [r_{2}].
  \)
\end{itemize}

\paragraph{Step 2: replacing constraints of $I$.} Initially, \(I\) is a \(d\)-regular \(\CSP(\TLin_{\mathsf{G}})\)-instance on \(V_{\orig}=[n]\times [3]\). For each original constraint $\big((v_{1},v_{2},v_{3}),S^{\mathsf{G}}_{b}\big)$ in $I$, we would like to replace it with a copy of the $\CSP(\Gamma\cup\Const_{D})$-instance $\cI_{b}$, which simulates $R_{b}$. This copied instance cannot stand alone but must be incorporated as a part of $\mathfrak{T}[I]$. In Step 3 below, we describe how to assign its variables to \(V_{\mathrm{const}}\) and \(V_{\mathrm{aux}}^{(1)}\).

\paragraph{Step 3: variable accommodation.} Suppose \(C=\big((v_{1},v_{2},v_{3}),S^{\mathsf{G}}_{b}\big)\) is the \(i\)-th constraint incident to \(v_{1}\) in \(I\). By assumption, \(\cI_{b}\) has at most \(r_{1}\) secondary variables and \(r_{1}\) constraints. Suppose $v_{1}=(j,1)$, where $j\in [n]$.
We allocate:
\[
\{j\}\times [r_{1}]\times D\subseteq V_{\mathrm{const}}\quad\text{and}\quad
\{i\}\times\{j\}\times[r_{1}]\subseteq V_{\mathrm{aux}}^{(1)}
\]
for these secondary variables.

Variables in \(\cI_{b}\) subject to constant constraints\footnote{These do not include the primary variables \(v_{1},v_{2},v_{3}\), since \(\cI_{b}\) simulates \(R\) on these variables and \(R\) does not fix any of them.} are identified with variables in \(\{j\}\times [r_{1}]\times D\subseteq V_{\mathrm{const}}\), where the choices in the third component correspond to which constants the variables are fixed to (the choices in the second component may be arbitrary). Other variables of \(\cI_{b}\) (neither primary nor constant-constrained) are represented by variables in \(\{i\}\times\{j\}\times[r_{1}]\subseteq V_{\mathrm{aux}}^{(1)}\). 

Having specified the variable accommodations, we can now place a copy of $\cI_{b}^{\natural}$ in $\mathfrak{T}[I]$, in replacement of the constraint $C=\big((v_{1},v_{2},v_{3}),S^{\mathsf{G}}_{b}\big)$ in $I$. Note that here we put in $\mathfrak{T}[I]$ only the non-constant constraints of $\cI_{b}$; the constant constraints will be ``simulated'' by the addition of an expander gadget in Steps 4 and 5. The set of constraints in this copy of $\cI_{b}^{\natural}$, as a subset of the constraint set of the eventual transformed instance $\mathfrak{T}[I]$, will be denoted by $\cQ_{C}$.

\paragraph{Step 4: expander gadget.} The constant-variable set \(V_{\mathrm{const}}=[n]\times [r_{1}]\times D\) can be viewed as a \(|D|\)-partite vertex set, with each part containing \(r_{1}n\) vertices (partitioned according to the third component). We may therefore use the expander gadget $H\in \supp\left(\cH_{D,r_{1}n,\ell}\right)$ (prepared in \Cref{subsubsec:prepatory}), taking $V_{\const}$ as its vertex set. For each hyperedge \(\mathbf{v}=(v_{x})_{x\in D}\) of \(H\), we would like to add an \(\End_{\Gamma}\)-constraint \((\mathbf{v},\End_{\Gamma})\) to \(\mathfrak{T}[I]\) (see \Cref{def:End_relation}).  
In the language of \Cref{def:I_of_RH}, this corresponds to placing a copy of \(\mathcal{I}[\End_{\Gamma},H]\) on \(V_{\mathrm{const}}\). However, since the relation \(\End_{\Gamma}\) may not belong to \(\Gamma\) itself, we must use the gadget instance \(\cI_{\End}\) to simulate \(\End_{\Gamma}\); this may introduce new variables, which must also be accommodated.

\paragraph{Step 5: variable accommodation (again).} Fix some \(z\in D\) as a “special element.”  
Suppose \(\mathbf{v}=(v_{x})_{x\in D}\) is the \(i\)-th hyperedge incident to \(v_{z}=(j,t,z)\) in \(H\).  
We create a copy of \(\cI_{\End}\) simulating \(\End_{\Gamma}\) on \(\mathbf{v}\), and allocate  
\[
\{i\}\times \{j\}\times \{t\}\times [r_{2}]\subseteq V_{\mathrm{aux}}^{(2)}
\]
for the secondary variables in this copy of \(\cI_{\End}\). The set of constraints in this copy of $\cI_{\End}$, as a subset of the constraint set of the eventual transformed instance $\mathfrak{T}[I]$, will be denoted by $\cQ_{\bfv}$.  

We repeat the process for every hyperedge \(\mathbf{v}\) of \(H\), thereby effectively placing a copy of \(\mathcal{I}[\End_{\Gamma},H]\) on \(V_{\mathrm{const}}\).

\paragraph{Conclusion.} We have now completed the construction of \(\mathfrak{T}[I]\).  
The constraint set of the final instance \(\mathfrak{T}[I]\) is the disjoint union of $\cQ_{C}$'s and $\cQ_{\bfv}$'s, where $C$ ranges in the constraint set of the original \(\CSP(\TLin_{\mathsf{G}})\)-instance \(I\) and $\bfv$ ranges in the hyperedge set of the expander gadget $H$. There are $dn$ constraints in $I$ and $r_{1}\ell n$ hyperedges in $H$; each $\cQ_{C}$ contains at most $r_{1}$ constraints and each $\cQ_{\bfv}$ contains at most $r_{2}$ constraints. Consequently, the total number of constraints in \(\mathfrak{T}[I]\) is at most  
\[
r_{1}\cdot dn+r_{2}\cdot r_{1}\ell n = (d+\ell r_{2})r_{1}n,
\]  
and at least \(dn+\ell r_{1}n\). To obtain degree bounds\footnote{The degree of a variable is the number of constraints it participates in.} on variables, we deal with the 4 parts of variable set separately:

\begin{itemize}
\item Every variable in $V_{\orig}$ is involved in at most $d$ of the sets $\cQ_{C}$'s, and so has degree at most $dr_{1}$.
\item Every variable $(j,t,x)\in V_{\const}$ is involved in at most $d$ of the sets $\cQ_{C}$'s (namely, for those $C$ that involve the variable $(j,1)\in [n]\times [3]$) and at most $\ell$ of the sets $\cQ_{\bfv}$'s. Therefore, the variable $(j,t,x)$ has degree at most $dr_{1}+\ell r_{2}$.

\item Every variable in $V_{\aux}^{(1)}$ is involved in at most 1 of the sets $\cQ_{C}$'s, and so has degree at most $r_{1}$.

\item Every variable in $V_{\aux}^{(2)}$ is involved in at most 1 of the sets $\cQ_{\bfv}$'s, and so has degree at most $r_{2}$.
\end{itemize}
Therefore, every variable in the instance $\mathfrak{T}[I]$ has degree at most $dr_{1}+\ell r_{2}$.

\subsection{Analysis of the Construction}\label{subsec:analysis}

Having completed the construction of the reduction map, we next prove its completeness and soundness.

\begin{lemma}[Completeness]\label{lem:completeness}
For any $\CSP\big(\TLin_{\mathsf{G}},[n]\times [3]\big)$-instance $I$ in the support of $\Dyes^{\mathsf{G},n,d}$, the transformed instance $\mathfrak{T}[I]\in \CSP(\Gamma)$ is satisfiable.
\end{lemma}

\begin{proof}
By \Cref{def:Dyes}, there exists an assignment $\tau:[n]\times [3]\rightarrow \mathsf{G}$ that satisfies $I$. We provide an assignment 
\[\widetilde{\tau}:V_{\orig}\cup V_{\const}\cup V_{\aux}^{(1)}\cup V_{\aux}^{(2)}\longrightarrow D\] 
that satisfies $\mathfrak{T}[I]$, as follows.

\begin{enumerate}[label=(\arabic*)]
\item For any original variable $v\in V_{\orig}=[n]\times [3]$, we pick $\widetilde{\tau}(v)$ to be an arbitrary element in the inverse image $\varphi^{-1}(\tau(v))$. Such an element exists because $\varphi$ is surjective.

\item For any constant variable $(j,t,x)\in V_{\const}=[n]\times [r_{1}]\times D$, we assign $\widetilde{\tau}\big((j,t,x)\big)=x$.
\item For each original constraint $C=\big((v_{1},v_{2},v_{3}),S^{\mathsf{G}}_{b}\big)$ in $I$, we know that there is an assignment to the instance $\cI_{b}$ (prepared in \Cref{subsubsec:prepatory}) that is consistent with $\tau$. We use this assignment to define \(\widetilde{\tau}(v)\) for every \(v\in V_{\mathrm{aux}}^{(1)}\) involved in $\cQ_{C}$.  

Any remaining variables in \(V_{\mathrm{aux}}^{(1)}\) not involved in any $\cQ_{C}$ may be assigned arbitrarily by \(\widetilde{\tau}\).

\item Since the $\CSP(\Gamma)$-instance $\cI_{\End}$ (prepared in \Cref{subsubsec:prepatory}) simulates the endomorphism relation of $\Gamma$, and since the identity map $\mathrm{id}:D\rightarrow D$ is an endomorphism, there is an assignment to $\cI_{\End}$ that restricts to the identity assignment on the primary variables. For each hyperedge $\bfv$ of $H$, we use this assignment to define $\widetilde{\tau}(v)$ for every $v\in V_{\mathrm{\aux}}^{(2)}$ involved in $\cQ_{\bfv}$. 

Any remaining variables in \(V_{\mathrm{aux}}^{(2)}\) not involved in any $\cQ_{\bfv}$ may be assigned arbitrarily by \(\widetilde{\tau}\).
\end{enumerate}

Items (1) and (3) above together ensure that $\widetilde{\tau}$ satisfies all constraints in the sets $\cQ_{C}$'s. Due to item (2), $\widetilde{\tau}$ restricts to the identity assignment on any hyperedge $\bfv$ of $H$, and item (4) then ensures that $\widetilde{\tau}$ also satisfies all constraints in the sets $\cQ_{\bfv}$'s. Therefore, $\widetilde{\tau}$ satisfies all constraints in $\mathfrak{T}[I]$.
\end{proof}

\begin{lemma}[Soundness]\label{lem:soundness}
For a random instance $\CSP\big(\TLin_{\mathsf{G}},[n]\times [3]\big)$-instance $I\sim \Dno^{\mathsf{G},n,d}$, the transformed instance $\mathfrak{T}[I]\in \CSP(\Gamma)$ has value at most $1-(10(d+\ell r_{2})r_{1}|D|)^{-1}$ with probability at least $1-o(1)$. Here, $o(1)$ denotes a term tending to 0 as $n\rightarrow\infty$.
\end{lemma}

\begin{proof}
As guaranteed in \Cref{subsubsec:prepatory}, the instance $I$ has value at most $2/3$ with probability $1-o(1)$. We claim that whenever $\val_{I}\leq 2/3$, any assignment must violate at least $n/(10|D|)$ constraints in $\mathfrak{T}[I]$. Since the total number of constraints in $\mathfrak{T}[I]$ is at most $(d+\ell r_{2})r_{1}n$, this implies that 
\[\val_{\mathfrak{T}[I]}\leq 1-\frac{n/(10|D|)}{(d+\ell r_{2})r_{1}n}=1-\frac{1}{10 (d+\ell r_{2})r_{1}|D|},\]
as desired.

Assume on the contrary that $\val_{I}\leq 2/3$ and there is an assignment 
\[\widetilde{\tau}:V_{\orig}\cup V_{\const}\cup V_{\aux}^{(1)}\cup V_{\aux}^{(2)}\longrightarrow D\] 
violating less than $n/(10|D|)$ constraints of $\mathfrak{T}[I]$. Since $H$ has $\ell r_{1}n$ hyperedges $\bfv$, at most $1/(10r_1|D|)$ fraction of the sets $\cQ_{\bfv}$'s contain constraints violated by $\widetilde{\tau}$. Therefore, $\widetilde{\tau}$ violates at most $1/(10 r_{1}|D|)$ fraction of the constraints $(\bfv,\End_{\Gamma})$ simulated by $\cQ_{\bfv}$'s. In other words, we have
\[\val_{\cI(\End_{\Gamma},H)}\big(\widetilde{\tau}_{\,|V_{\const}}\big)\geq 1-\frac{1}{10 r_{1}|D|}.\]
Since $H$ satisfies the property stated in \Cref{lem:expander_gadget} (with respect to the sub-unique relation $\End_{\Gamma}$), it follows from the conclusion of \Cref{lem:expander_gadget} that there exists a relational endomorphism $\psi:D\rightarrow D$ of $(D,\Gamma)$ such that 
\begin{equation}\label{eq:errors_on_constant}
\Pru{(j,t,x)\in V_{\const}}{\widetilde{\tau}\big((j,t,x)\big)\neq \psi(x)}\leq \frac{1}{5 r_{1}|D|}.
\end{equation}

As $(D,\Gamma)$ is a core, $\psi^{-1}$ is also a relational endomorphism. Thus the assignment

\[\psi^{-1}\circ \widetilde{\tau}:V_{\orig}\cup V_{\const}\cup V_{\aux}^{(1)}\cup V_{\aux}^{(2)}\longrightarrow D\]
satisfies the same subset of constraints in $\mathfrak{T}[I]$ as $\widetilde{\tau}$ does. In particular, it violates at most $n/(10|D|)$ of these constraints. 

Now consider the assignment $\tau:[n]\times [3]\rightarrow \mathsf{G}$ defined as follows. For any $v\in V_{\orig}=[n]\times [3]$, let 
\[
\tau(v)=\begin{cases}
\varphi\big(\psi^{-1}(\widetilde{\tau}(v))\big), &\text{if }\psi^{-1}\big(\widetilde{\tau}(v)\big)\in D',\\
0 ,&\text{otherwise}.
\end{cases}
\]
It is easy to see that any original constraint $C=\big((v_{1},v_{2},v_{3}),S^{\mathsf{G}}_{b}\big)$ of $I$ is violated by $\tau$ only if either of the following happens:

\begin{enumerate}[label=(\arabic*)]
\item The assignment $\psi^{-1}\circ \widetilde{\tau}$ violates one of the constraints in $\cQ_{C}$.
\item For some constant variable $v=(j,t,x)\in V_{\const}$ involved in $\cQ_{C}$, it happens that $\psi^{-1}(\widetilde{\tau}(v))\neq x$.
\end{enumerate}

The second scenario above happens for at most $|V_{\const}|/(5r_{1}|D|)=n/5$ constant variables $v\in V_{\const}$, due to \eqref{eq:errors_on_constant}. Since every constant variable in $ V_{\const}$ is involved in at most $d$ of the sets $\cQ_{C}$'s, it follows that the second scenario happens for at most $dn/5$ constraints of $I$. 

The first scenario happens for at most $n/(10|D|)$ constraints of $I$, since no two $\cQ_{C}$'s share constraints in $\mathfrak{T}[I]$, and at most $n/(10|D|)$ constraints in $\mathfrak{T}[I]$ are violated by $\psi^{-1}\circ \widetilde{\tau}$.

In total, at most 
\[
\frac{dn}{5}+\frac{n}{10|D|}<\frac{dn}{3}
\]
constraints of $I$ are violated by $\tau$, contradicting the assumption $\val_{I}\leq 2/3$. 
\end{proof}

The following lemma proves that the indistinguishability between $\Dyes$ and $\Dno$, as established in \Cref{lem:3LIN_indistinguishable}, is ``preserved'' by the reduction map $\mathfrak{T}[\cdot]$. 

\begin{lemma}[Indistinguishability]\label{lem:indistinguishability}
Let $V = V_{\orig} \cup V_{\const} \cup V_{\aux}^{(1)} \cup V_{\aux}^{(2)}$.  
There exists a constant $\gamma > 0$, depending only on $(D,\Gamma)$ but not on $n$, such that any algorithm $\cA$ computing a function  
\[
\cA : \CSP(\Gamma,V) \longrightarrow \{0,1\}
\]  
using at most $\gamma n$ variable-queries satisfies  
\[
\left|\Pru{I \sim \Dyes^{\mathsf{G},n,d}}{\cA(\mathfrak{T}[I])=1\big.} - 
\Pru{I \sim \Dno^{\mathsf{G},n,d}}{\cA(\mathfrak{T}[I])=1\big.}\right|\leq o(1),
\]
where $o(1)$ denotes a term tending to 0 as $n\rightarrow \infty$.
\end{lemma}

\begin{proof}
Let $\cA$ be such an algorithm that makes at most $\gamma n$ variable-queries.  
We claim that there exists an algorithm $\cA'$ computing a function  
\[
\cA' : \supp\big(\Dyes^{\mathsf{G},n,d}\big) \cup \supp\big(\Dno^{\mathsf{G},n,d}\big) \longrightarrow \{0,1\}
\]  
using at most $\gamma n$ variable-queries, such that  
\begin{equation}\label{eq:A'_simulates_A}
\cA(\mathfrak{T}[I]) = \cA'(I) \quad 
\text{for all } I \in \supp\big(\Dyes^{\mathsf{G},n,d}\big) \cup \supp\big(\Dno^{\mathsf{G},n,d}\big).
\end{equation}
In other words, given an instance $I$ from $\Dyes$ or $\Dno$, the algorithm $\cA'$ effectively simulates the execution of $\cA$ on $\mathfrak{T}[I]$.

Consider the types of queries that $\cA$ may make on input $\mathfrak{T}[I]$:
\begin{itemize}
\item A query to a variable in $V_{\aux}^{(2)}$ reveals no information about $I$.  
\item A query to $(i,j,t)\in V_{\aux}^{(1)}$ reveals no more information than a query to $(j,1)\in  V_{\orig}$ in $I$.  
\item A query to $(j,t,x)\in V_{\const}$ reveals no more information than a query to $(j,1)\in V_{\orig}$ in $I$.  
\item A query to $v \in V_{\orig}$ reveals no more information than a query to the same variable $v$ in $I$.  
\end{itemize}

Thus every query to $\mathfrak{T}[I]$ performed by $\cA$ can be replaced by a query to $I$ by $\cA'$, proving the claim.  

Now let $\delta\in (0,1)$ and $\cE$ be the parameter and the event from \Cref{lem:3LIN_indistinguishable}, respectively. Conditioned on the event that the pair of instances $(I_{\yes},I_{\no})$ sampled from $\Dpair^{\mathsf{G},n,d}$ falls in $\cE$, the distributions of $I_{\yes}$ and $I_{\no}$ coincide when restricted to any $\delta n$ variables, according to \Cref{lem:3LIN_indistinguishable}. We set $\gamma = \delta/(3d+1)$. Since any algorithm $\cA'$ with $\gamma n$ queries to $I$ can observe constraints involving at most $\gamma n(1+3d)=\delta n$ variables in $V_{\orig}$, it follows that
\[
\Pru{(I_{\yes},I_{\no})\sim \Dpair^{\mathsf{G},n,d}}{\cA'(I_{\yes})=1\,\middle|\,\cE}=\Pru{(I_{\yes},I_{\no})\sim \Dpair^{\mathsf{G},n,d}}{\cA'(I_{\no})=1\,\middle|\,\cE}.
\]
Since $\Dpair^{\mathsf{G},n,d}(\cE)\geq 1-o(1)$, the fact that $\Dyes$ and $\Dno$ are the marginals of $\Dpair^{\mathsf{G},n,d}$ implies 
\[
\left|\Pru{I \sim \Dyes^{\mathsf{G},n,d}}{\cA'(I)=1\big.} - 
\Pru{I \sim \Dno^{\mathsf{G},n,d}}{\cA(I)=1\big.}\right|\leq o(1).
\]
Combined with \eqref{eq:A'_simulates_A}, this yields the desired conclusion.
\end{proof}

\subsection{Conclusion in the General Case}\label{subsec:general_conclusion}

We are now ready to conclude the proof of \Cref{thm:main_repetition}.

\begin{proof}[Proof of \Cref{thm:main_repetition}]
We first consider the case where $(D,\Gamma)$ is a core.  
In \Cref{subsec:construction_in_core}, we obtained a constant $d\in\mathbb{N}$, an Abelian group $\mathsf{G}$, and a reduction map
\[
\mathfrak{T}:\supp\big(\mathcal{D}_{\text{yes}}^{\mathsf{G},n,d}\big)\cup\supp\big(\mathcal{D}_{\text{no}}^{\mathsf{G},n,d}\big) \longrightarrow \CSP(\Gamma).
\]
The map $\mathfrak{T}$ produces instances over the variable set
\[
V_{\orig}\cup V_{\const}\cup V_{\aux}^{(1)}\cup V_{\aux}^{(2)},
\]
whose total size is $(3 + |D|r_{1} + dr_{1} + \ell r_{1}r_{2})n$.  
As noted at the end of \Cref{subsubsec:construction}, the instances produced by $\mathfrak{T}$ contain at least $(d+\ell r_{1})n$ constraints, and each variable has degree at most $dr_{1}+\ell r_{2}$.  
Define the constants
\[
\beta = 3 + |D|r_{1} + dr_{1} + \ell r_{1}r_{2}, 
\qquad 
\alpha = \frac{d+\ell r_{1}}{\beta},
\qquad 
d' = dr_{1}+\ell r_{2}.
\]
It follows that all instances produced by $\mathfrak{T}$ are valid in the $\BD^*(d', \alpha,\beta n)$ model.  
By combining \Cref{lem:completeness,lem:soundness,lem:indistinguishability}, there exists a constant $\gamma > 0$ such that, for sufficiently large $n$, the problem $\Mcsp{\Gamma}{1}{1-\varepsilon}$ requires at least $\gamma n$ queries in the $\BD^*(d', \alpha,\beta n)$ model, where
\[
\varepsilon := \frac{1}{10(d+\ell r_{2})r_{1}|D|}.
\]

We now turn to the case where $(D,\Gamma)$ is not a core.  
By \Cref{prop:core_existence}, $(D,\Gamma)$ is homomorphically equivalent to a repetition-closed core $(D',\Gamma')$.  
We claim that $\Mcsp{\Gamma'}{1}{1-\varepsilon}$ is essentially equivalent to $\Mcsp{\Gamma}{1}{1-\varepsilon}$.  

Let $\varphi_{1}:D\rightarrow D'$ and $\varphi_{2}:D'\rightarrow D$ be the corresponding relational homomorphisms.  
Since $\Gamma$ and $\Gamma'$ are bijectively matched and corresponding relations have the same arity, any $\CSP(\Gamma)$ instance $I=(V,\mathcal{C})$ can be viewed interchangeably as a $\CSP(\Gamma')$ instance $I'=(V,\mathcal{C}')$, and vice versa.  
Furthermore, for any assignment $\tau:V\rightarrow D$ to $I$, we have
\[
\val_{I}(\tau) = \val_{I'}(\varphi_{1}\circ \tau),
\]
by the definition of relational homomorphism.  
Conversely, for any assignment $\tau':V\rightarrow D'$, we have
\[
\val_{I}(\varphi_{2}\circ \tau') = \val_{I'}(\tau').
\]
We therefore have $\val_{I}=\val_{I'}$. 

By the above discussion, the identification between $\CSP(\Gamma)$ and $\CSP(\Gamma')$ yields trivial reductions between $\Mcsp{\Gamma'}{1}{1-\varepsilon}$ and $\Mcsp{\Gamma}{1}{1-\varepsilon}$.  
Since $(D',\Gamma')$ also has unbounded width (by \Cref{prop:homo_equivalence_unbounded_width}), the hardness of $\Mcsp{\Gamma}{1}{1-\varepsilon}$ follows from the established hardness of $\Mcsp{\Gamma'}{1}{1-\varepsilon}$.
\end{proof}

\section*{Acknowledgements}

The author would like to thank Dor Minzer, Ronitt Rubinfeld, Shuo Wang and Standa Živný for many stimulating discussions during the development of this work. Special thanks to Ronitt Rubinfeld, for her helpful suggestions in improving the presentation of the paper.

\addcontentsline{toc}{section}{References}
\bibliography{reference}
\bibliographystyle{alpha}

\appendix

\section{Definition of Bounded Width}\label{sec:bounded-width}

In this section, we provide the definition and basic properties of the ``bounded width'' notion for CSPs. The class of bounded-width CSP templates is defined by the correctness of an algorithm, which we call the width-$(k,\ell)$ algorithm, in deciding satisfiability of instances.

\begin{definition}[\cite{feder1998computational}]
Let $k,\ell$ be integers such that $\ell\geq k\geq 1$. A CSP template $(D,\Gamma)$ is said to have width $(k,\ell)$ if for all instances $I\in \CSP(\Gamma)$, \Cref{alg:width_k_ell} correctly decides whether $I$ is satisfiable. A template $(D,\Gamma)$ is said to have bounded width if it has width $(k,\ell)$ for some integers $k,\ell$; otherwise, we say $(D,\Gamma)$ has unbounded width.
\end{definition}
\begin{algorithm}
    \DontPrintSemicolon
    \SetKwInOut{Input}{Input}\SetKwInOut{Output}{Output}
    \caption{Width-$(k,\ell)$ Algorithm}\label{alg:width_k_ell}
    \Input{a $\CSP(\Gamma)$ instance $I=(V,(C_{1},\dots,C_{m}))$}
    \Output{\texttt{satisfiable} or \texttt{unsatisfiable}}
    \For{each subset $U\subseteq V$ of size $\ell$}{
        Initialize a set of partial assignments $S_{U}\gets D^{U}$\;
        \For{each pair $(\tau,i)$ of partial assignment $\tau\in D^{U}$ and constraint index $i\in [m]$}{
            \If{$U$ contains the scope of $C_{i}$ \textup{\textbf{and}} $\tau$ does not satisfy $C_{i}$}{
                $S_{U}\gets S_{U}\setminus \{\tau\}$\label{line:remove_inconsistencies_original}
                \tcp*{remove inconsistencies with original constraints}
            }
        }
    }
    \tcp*{finish initialization; start main procedure}
    \Repeat{the sets $S_{U}$ stop changing}{
    \For{each pair of subsets $U_{1},U_{2}\subseteq V$ such that $|U_{1}|=|U_{2}|=\ell$}{
    \For{each (if any) subset $W\subseteq U_{1}\cap U_{2}$ such that $|W|=k$}{
    \For{each $\tau\in S_{U_{1}}$}{
    \If{$\tau_{|W}\neq \sigma_{|W}$ for all $\sigma\in S_{U_{2}}$}{
    $S_{U_{1}}\gets S_{U_{1}}\setminus \{\tau\}$\tcp*{remove inconsistencies in any $k$ variables}\label{line:remove_inconsistencies_in_k}
    }
    }
    }
    }
    }
    \textbf{Output} \texttt{unsatisfiable} if $S_{U}=\emptyset$ for some $U$; otherwise \textbf{output} \texttt{satisfiable}
\end{algorithm}

Note that \Cref{alg:width_k_ell} is always correct if it outputs \texttt{unsatisfiable}, so the only mistake is when the input instance is unsatisfiable but it outputs \texttt{satisfiable}. Also, since the running time of \Cref{alg:width_k_ell} is clearly polynomial in $|V|$ and $m$, we immediately have the following property:

\begin{proposition}\label{prop:bounded-width_in_P}
For a CSP template $(D,\Gamma)$, we let $\Gamma\textsc{-Sat}$ be the problem of deciding satisfiability of instances from $\CSP(\Gamma)$. If $(D,\Gamma)$ has bounded width, then $\Gamma\textsc{-Sat}\in \mathbf{P}$.
\end{proposition}

The following basic property of bounded-width templates is used in \Cref{subsec:general_conclusion}.

\begin{proposition}\label{prop:homo_equivalence_unbounded_width}
Let $(D_{1},\Gamma_{1})$ and $(D_{2},\Gamma_{2})$ be two relational structures that are homomorphically equivalent to each other. Then $(D_{1},\Gamma_{1})$ has bounded width if and only if $(D_2,\Gamma_{2})$ has bounded width.
\end{proposition}
\begin{proof}
Suppose $\varphi_{1}\colon D_{1}\to D_{2}$ and $\varphi_{2}\colon D_{2}\to D_{1}$ are homomorphisms.  
Assume that $(D_{2},\Gamma_{2})$ has width $(k,\ell)$.  
To show that $(D_{1},\Gamma_{1})$ also has width $(k,\ell)$, it suffices to prove that for any instance 
$I\in\CSP(\Gamma_{1})$, if \Cref{alg:width_k_ell} outputs \texttt{satisfiable} on $I$, then $I$ is indeed satisfiable.  

For local assignments we introduce the following notation:

\begin{enumerate}[label=(\arabic*)]
\item For each subset $U\subseteq V$ and each collection of local assignments 
$S\subseteq D_{1}^{U}$, set
\[
\varphi_{1}^{U}(S)
:=\bigl\{\varphi_{1}\circ\tau \,\bigm|\, \tau\in S\bigr\}
\subseteq D_{2}^{U}.
\]

\item For any $W,U\subseteq V$ with $W\subseteq U$, and each set of local assignments 
$S\subseteq D_{1}^{U}$, define the restriction
\[
\Res_{U,W}(S)
:=\bigl\{\tau_{|W}\,\bigm|\,\tau\in S\bigr\}.
\]
By slight abuse of notation we use the same symbol 
$\Res_{U,W}(S)$ when $S\subseteq D_{2}^{U}$.
\end{enumerate}

Let $\{S_{U}\}_{U\subseteq V,\,|U|=\ell}$ denote the family of local assignment sets produced by 
\Cref{alg:width_k_ell} at the end of its execution on $I$.  
Since we have assumed that \Cref{alg:width_k_ell} outputs \texttt{satisfiable} on $I$, it follows that $S_{U}\neq\emptyset$ for every $U$.  
By \Cref{line:remove_inconsistencies_in_k} of \Cref{alg:width_k_ell}, we also have 
\[
\Res_{U_{1},W}(S_{U_{1}})=\Res_{U_{2},W}(S_{U_{2}})
\]
for any $k$-element set $W\subseteq V$ and any $\ell$-element sets $U_{1},U_{2}\subseteq V$ containing $W$.  
Furthermore, by \Cref{line:remove_inconsistencies_original}, for any $\ell$-element set $U\subseteq V$, every assignment in $S_{U}$ satisfies all constraints of $I$ whose scopes are contained in $U$.

Construct an instance $I'\in\CSP(\Gamma_{2})$ by replacing each constraint of $I$ with the corresponding constraint from $\Gamma_{2}$.  
Since $\varphi_{1}$ is a relational homomorphism, \Cref{def:relational_homomorphism} implies that for any $\ell$-element set $U\subseteq V$, each assignment 
$\tau'\in \varphi_{1}^{U}(S_{U})$ satisfies all constraints of $I'$ whose scopes are contained in $U$.  
Moreover, for any $k$-element set $W\subseteq V$ and $\ell$-element sets $U_{1},U_{2}\subseteq V$ containing $W$, we have
\[
\Res_{U_{1},W}\bigl(\varphi_{1}^{U_{1}}(S_{U_{1}})\bigr)
=\varphi_{1}^{W}\bigl(\Res_{U_{1},W}(S_{U_{1}})\bigr)
=\varphi_{1}^{W}\bigl(\Res_{U_{2},W}(S_{U_{2}})\bigr)
=\Res_{U_{2},W}\bigl(\varphi_{1}^{U_{2}}(S_{U_{2}})\bigr).
\]
Thus the (nonempty) family of local assignment sets 
$\bigl\{\varphi_{1}^{U}(S_{U})\bigr\}_{U\subseteq V,\,|U|=\ell}$ persists to the end of \Cref{alg:width_k_ell} when run on $I'$.  
Since $(D_{2},\Gamma_{2})$ has width $(k,\ell)$ by assumption, $I'$ is satisfiable.  
Finally, for any global assignment $\tau'\colon V\to D_{2}$ satisfying $I'$, the composed assignment 
$\varphi_{2}\circ \tau'\colon V\to D_{1}$ satisfies $I$ because $\varphi_{2}$ is also a relational homomorphism.  
Hence $I$ is satisfiable, completing the proof.
\end{proof}

The following simple observation is used in the proof of \Cref{lem:main_algebra}.

\begin{proposition}\label{prop:Const_preserves_unbounded_width}
Let $(D,\Gamma_{1})$ and $(D,\Gamma_{2})$ be two relational structures on a same domain. If $(D,\Gamma_{1})$ has unbounded width and $\Gamma_{1}\subseteq \Gamma_{2}$, then $(D,\Gamma_{2})$ also has unbounded width.
\end{proposition}

\begin{proof}
This is simply because any $\CSP(\Gamma_{1})$ instance is also a $\CSP(\Gamma_{2})$ instance. If \Cref{alg:width_k_ell} makes error on some instance $I\in\CSP(\Gamma_{1})$, then it also makes error on the same instance $I\in \CSP(\Gamma_{2})$. 
\end{proof}

\begin{remark}\label{rem:colorability_unbounded_width}
A decidable characterization of bounded-width CSP templates was given by Barto \cite[Corollary 8.5]{barto2014collapse}. Using Barto's criterion, it is easy to verify that for any $k,\ell\geq 2$ with $(k,\ell)\neq (2,2)$, the CSP template expressing $k$-colorability of $\ell$-uniform hypergraphs does not have bounded width.
\end{remark}

\section{Allowing Variable Repetition}\label{sec:allow_repeition}

\begin{proof}[Proof of \Cref{thm:main_rephrase} assuming \Cref{thm:main_repetition}]
Given any relational structure $(D,\Gamma)$, we will define a reduction map $\mathfrak{R}:\CSP(\overline{\Gamma})\rightarrow \CSP(\Gamma)$. Let $k$ be the maximum arity of any relation in $\Gamma$. Given an instance $I=(V,\cC)\in \CSP(\overline{\Gamma})$, we construct $\mathfrak{R}[I]$ as follows.
\begin{enumerate}[label=(\arabic*)]
\item Replace every variable $v\in V$ by $m:=k|D|$ copies $v^{(1)},\dots,v^{(m)}$. The new variable set is denoted by $V\times [m]$.
\item Each constraint $C\in \cC$ can be viewed as a requirement \[R(\tau(v_{1}),\dots,\tau(v_{r}))=1\]
for assignments $\tau:V\rightarrow D$, where $R\in \Gamma$ and $v_{1},\dots,v_{r}$ are \emph{not necessarily distinct} variables in $V$. In place of this constraint, we add to $\mathfrak{R}[I]$ every constraint of the form
\[
\left(\big(v_{1}^{(i_{1})},\dots,v_{r}^{(i_{r})}\big),R\right)
\]
such that the chosen variables $v_{1}^{(i_{1})},\dots,v_{r}^{(i_{r})}$ are pairwise distinct. 
\end{enumerate}
The reduction map has the following desired properties.
\begin{enumerate}[label=(\arabic*)]
\item Note that every constraint of $I$ corresponds to at least 1 and at most $m^{k}$ constraints of $\mathfrak{R}[I]$. If $I$ has maximum degree at most $d$, then $\mathfrak{R}[I]$ has maximum degree at most $m^{k}d$. 
\item It is clear that if $I$ is satisfiable then $\mathfrak{R}[I]$ is also satisfiable. 
\item If $\val_{I}\leq 1-\varepsilon$, then we claim that $\val_{\mathfrak{R}[I]}\leq 1-m^{-k}\varepsilon$. Given any assignment $\widetilde{\tau}:V\times [m]\rightarrow D$ and any $v\in V$, by the pigeonhole principle some element of $D$ appears in $\tau(v^{(1)}),\dots,\tau(v^{(m)})$ at least $k$ times. We define $\tau(v)$ to be any such element. For any constraint $C\in \cC$ that this assignment $\tau:V\rightarrow D$ violates, at least one of the (at most $m^{k}$) constraints in $\mathfrak{R}[I]$ corresponding to $C$ is violated by $\widetilde{\tau}$. Therefore, the assumption that $\val_{I}(\tau)\leq 1-\varepsilon$ implies $\val_{\mathfrak{R}[I]}(\widetilde{\tau})\leq 1-m^{-k}\varepsilon$.
\item Finally, it is easy to see that every query to a variable $v^{(i)}$ in $\mathfrak{R}[I]$ reveals no more information than a query to $v$ in $I$. 
\end{enumerate}

Now suppose \Cref{thm:main_repetition} holds. Then for any unbounded-width CSP template $(D,\Gamma)$, since $(D,\overline{\Gamma})$ is repetition-closed and also has unbounded width (due to \Cref{prop:Const_preserves_unbounded_width}), there exists $\varepsilon>0$ such that $\Mcsp{\overline{\Gamma}}{1}{1-\varepsilon}$ requires $\Omega(n)$ queries in the $\BD^*(d,\alpha,n)$ model. Due to the reduction above, this implies that $\Mcsp{\Gamma}{1}{1-m^{-k}\varepsilon}$ requires $\Omega(n)$ queries in the \(\BD^*(m^{k}d,\alpha ,mn)\) model, where $k$ is the largest arity of any relation in $\Gamma$ and $m:=k|D|$.
\end{proof}
\section{Galois Duality Proof}\label{sec:Galois_proof}

\begin{proof}[Proof of \Cref{thm:Galois}]

The proof proceeds in the following steps.

\paragraph{Step 1: the variable pool.} Recall from \Cref{def:relational_clone} that to show \(R\in\clo{\Gamma}\) we must build a \(\CSP(\Gamma)\) instance that “simulates” \(R\).  
Such an instance contains \(k\) primary variables \(v_{1},\dots,v_{k}\) on which \(R\) is simulated, together with a number of secondary variables.  
Here, however, we first introduce a single \emph{pool} of variables containing both the primary and secondary variables we need.  
Only afterwards will we specify which are the primary variables.
 
Let \(R^{-1}(1)=\{x=(x_{1},\dots,x_{k})\in D^{k}\mid R(x)=1\}\) and put \(m=|R^{-1}(1)|\).  
We create \(|D|^{m}\) variables \(u_{y}\), indexed by all vectors \(y\in D^{m}\).  
Initially, the instance \(I\) on variables \((u_{y})_{y\in D^{m}}\) has no constraints; the next step explains how constraints are added.

\paragraph{Step 2: encoding “polymorphism”.}  
The key idea is that the condition of being a polymorphism of \(\Gamma\) can be written syntactically as a CSP instance.  
An \(m\)-ary operation \(f\colon D^{m}\to D\) corresponds to an assignment  
\[
\tau:\{u_{y}\mid y\in D^{m}\}\to D,
\]
interpreting \(\tau(u_{y})=f(y)\). Recall that for $f$ to qualify as a polymorphism of $(D,\Gamma)$, it needs to pass all tests given by a relation $R'\in \Gamma$ and a matrix of elements $(x_{i,j})\in D^{\ell\times m}$, where $\ell$ is the arity of $R'$, such that each column of the matrix satisfies $R'$. In the test, one checks the output of $f$ on each \emph{row} of the matrix and verifies if the combined outputs as an element of $D^{\ell}$ satisfies $R'$. We now describe a procedure producing a CSP instance that enforces these tests.

\begin{itemize}
\item For each test \((R',(x_{i,j}))\in\Gamma\times D^{\ell\times m}\):
  \begin{enumerate}
  \item For each row index \(i\in[\ell]\), define \(y^{(i)}\in D^{m}\) by \(y^{(i)}_{j}=x_{i,j}\) for all \(j\in[m]\).
  \item Add to \(I\) the constraint
  \[
  \bigl((u_{y^{(1)}},\dots,u_{y^{(\ell)}}),R'\bigr).
  \]
  \end{enumerate}
\end{itemize}
By construction, an assignment \(\tau\) satisfies \(I\) iff the corresponding operation \(f\) is a polymorphism of \(\Gamma\).
 
\paragraph{Step 3: selecting the primary variables.}  
Let \(w^{(1)},\dots,w^{(m)}\) be an enumeration of elements of \(R^{-1}(1)\).  
For each \(i\in[k]\) define \(z^{(i)}\in D^{m}\) by \(z^{(i)}_{j}=w^{(j)}_{i}\) for \(j\in[m]\).  
The condition that \(R\) is irredundant now easily translates to (desired) property that the vectors \(z^{(i)}\) are distinct.  
We designate as the primary variables
\[
\bigl(u_{z^{(1)}},\dots,u_{z^{(k)}}\bigr).
\]
We have now completed the construction of $I$.

\paragraph{Step 4: correctness.}  
For any \(x=(x_{1},\dots,x_{k})\in R^{-1}(1)\) we have \(x=w^{(j)}\) for some \(j\in[m]\).  
The assignment \(\tau(u_{y})=y_{j}\) (corresponding to the projection polymorphism \(\pi_{j}\colon D^{m}\to D\)) then satisfies \(I\) and yields
\(\tau(u_{z^{(i)}})=x_{i}\) for all \(i\in[k]\).

Conversely, if \(\tau\) satisfies \(I\), the induced operation \(f:D^{m}\to D\) is a polymorphism of \((D,\Gamma)\) and hence preserves \(R\).  
Since each column of the matrix \((w^{(j)}_{i})_{i\in[k],j\in[m]}\in D^{k\times m}\) satisfies \(R\), the tuple of row outputs
\[
\bigl(f(z^{(1)}),\dots,f(z^{(k)})\bigr)
=\bigl(\tau(u_{z^{(1)}}),\dots,\tau(u_{z^{(k)}})\bigr)\in D^{k}
\]
also satisfies \(R\).  
Thus the instance \(I\in\CSP(\Gamma)\) correctly simulates \(R\) on the designated primary variables in the sense of \Cref{def:relational_clone}, which proves \(R\in\clo{\Gamma}\).
\end{proof}

\section{Random Regular Hypergraphs}\label{sec:random_regular_hypergraph}

In this section, we outline the proofs of \Cref{lem:3LIN_exists_d,lem:3LIN_indistinguishable,lem:expander_gadget}, all of which revolve around random regular hypergraphs. The first two lemmas, \Cref{lem:3LIN_exists_d,lem:3LIN_indistinguishable}, are only minor adaptations of results from \cite{bogdanov2002lower}, and our proofs closely follow those in that work. In contrast, \Cref{lem:expander_gadget} is new to this paper, though its proof uses arguments that are largely analogous to those used in the proof of \Cref{lem:3LIN_exists_d}.

\subsection{Random Regular Hypergraphs are Locally Sparse}

We first sketch a proof of \Cref{lem:3LIN_indistinguishable}. 

Recall from \Cref{def:perfect_matching} that $\cM_{[3],n}$ denotes the set of all $3$-partite perfect matchings on the vertex set $[n]\times [3]$. For any matching $M\in \cM_{[3],n}$ and any vertex subset $U\subseteq [n]\times [3]$, we let $M[U]$ denote the set of hyperedges $(v_{1},v_{2},v_{3})\in M$ such that $v_{1},v_{2},v_{3}\in U$. The following lemma is a slightly modified version of \cite[Lemma 6]{bogdanov2002lower}. 

\begin{lemma}[\cite{bogdanov2002lower}]\label{lem:local_sparse}
    For any fixed positive integer $d$, there exists a constant $\delta\in (0,1)$ such that if $M^{(1)},\dots,M^{(d)}$ are perfect matchings sampled independent and uniformly from $\cM_{[3],n}$, then the condition
    \begin{equation}\label{eq:locally_sparse}
    \sum_{i=1}^{d}\left|M^{(i)}[U]\right|<\frac{2}{3}|U|\qquad \text{for all }U\subseteq [n]\times [3]\text{ with }1\leq |U|\leq \delta n
    \end{equation}
    is satisfied with probability $1-o(1)$, where $o(1)$ is a term tending to 0 as $n\rightarrow \infty$.
\end{lemma}
\begin{proof}
For a fixed $U\subseteq [n]\times [3]$ and any $i\in [3]$, let $U_{i}=\{(j,i)\in U\mid j\in [n]\}$ be the set of vertices in $U$ that belong to the $i$-th part in the 3-way partition of $[n]\times [3]$. Let $W$ be the set of pairs $(j,i)\in U_{1}\times [d]$ such that some hyperedge $(v_{1},v_{2},v_{3})$ in $M^{(i)}$ contains the vertex $v_{1}=(j,1)$ and satisfies $v_{1},v_{2},v_{3}\in U$. We therefore have
\[
|W|=\sum_{i=1}^{d}\left|M^{(i)}[U]\right|.
\]
For any given subset $W_{0}\subseteq U_{1}\times [d]$, it is not hard to see that
\[
\Pru{M^{(1)},\dots,M^{(d)}\in \cM_{[3],n}}{W\supseteq W_{0}}\leq \left(\frac{|U_{2}|\times |U_{3}|}{n^{2}}\right)^{|W_{0}|}.
\]
Therefore we have
\begin{align*}
\Pru{M^{(1)},\dots,M^{(d)}\in \cM_{[3],n}}{|W|\geq \frac{2}{3}|U|}&\leq \sum_{\substack{W_{0}\subseteq U_{1}\times [d]\\ |W_{0}|=\lceil 2|U|/3\rceil}}\Pru{M^{(1)},\dots,M^{(d)}\in \cM_{[3],n}}{W\supseteq W_{0}}\\
&\leq \binom{d|U_{1}|}{\lceil2|U|/3\rceil}\left(\frac{|U_{2}|\times |U_{3}|}{n^{2}}\right)^{\lceil2|U|/3\rceil}\leq \left(\frac{5d|U|^{2}}{n^{2}}\right)^{\lceil 2|U|/3\rceil}.
\end{align*}
Pick $\delta= 10^{-6}d^{-2}$. By taking a union bound over all $U\subseteq [n]\times [3]$ with $1\leq |U|\leq \delta n$, we have
\[
\Pr{\text{condition \eqref{eq:locally_sparse} fails}\big.}\leq \sum_{t=1}^{\lfloor \delta n\rfloor }\binom{3n}{t}\left(\frac{5dt^{2}}{n^{2}}\right)^{2t/3}\leq \sum_{t=1}^{\lfloor \delta n\rfloor}\left(\frac{10^{5}t}{n}\right)^{t/3}\leq o(1).\qedhere
\]
\end{proof}

As stated in \cite{bogdanov2002lower}, the conclusion of \Cref{lem:local_sparse} still holds if the constant $2/3$ in \eqref{eq:locally_sparse} is replaced by any constant greater than $1/2$. The reason for choosing $2/3$ in \eqref{eq:locally_sparse} is that it is the largest constant for which the following observation remains true. 

\begin{proposition}\label{prop:locally_sparse}
Suppose $M^{(1)}, \dots, M^{(d)}$ are perfect matchings in $\cM_{[3],n}$ that satisfy \eqref{eq:locally_sparse}. Then, for any subset $U \subseteq [n] \times [3]$ with $|U| \leq \delta n$, the following hold:
\begin{enumerate}[label=(\arabic*)]
\item The edge sets $M^{(i)}$, for $i \in [d]$, are pairwise disjoint.
\item There exists a total order on the union $\bigcup_{i=1}^{d} M^{(i)}$ such that every hyperedge in this union is incident to at least one vertex that does not appear in any hyperedge preceding it in the order.
\end{enumerate}
\end{proposition}
\begin{proof}
The first item holds because no hyperedge can appear in more than one of the $M^{(i)}$'s; otherwise, there would exist a set $U$ of three vertices such that $\sum_{i=1}^{d}\left|M^{(i)}[U]\right| \ge 2$, violating condition~\eqref{eq:locally_sparse}. 

We now prove item~(2). Let $U_{0} \subseteq U$ denote the set of vertices incident to at least one hyperedge in the union $E := \bigcup_{i=1}^{d} M^{(i)}$. Applying condition~\eqref{eq:locally_sparse} to $U_{0}$ implies that a uniformly random vertex from $U_{0}$ is, in expectation, incident to fewer than two hyperedges in $E$. Hence, there exists at least one vertex in $U_{0}$ that is incident to exactly one hyperedge in $E$. 

We remove this hyperedge from $E$ and update $U_{0}$ to be the set of vertices incident to at least one hyperedge in the remaining set $E$. Repeating this process iteratively, we eventually remove all hyperedges from $E$. This procedure defines a total order on the edge set $\bigcup_{i=1}^{d} M^{(i)}$, where edges removed later are considered smaller in the order. The desired conclusion then follows directly from the construction.
\end{proof}
We now use \Cref{lem:local_sparse} and \Cref{prop:locally_sparse} to prove \Cref{lem:3LIN_indistinguishable}.

\begin{proof}[Proof Sketch of \Cref{lem:3LIN_indistinguishable}]
Recall from \Cref{def:Dyes} that for any sample $(I_{\yes}, I_{\no}) \sim \Dpair^{\mathsf{G},n,d}$, the same collection of $d$ perfect matchings $M^{(1)}, \dots, M^{(d)} \in \cM_{[3],n}$ underlies both $I_{\yes}$ and $I_{\no}$. Applying \Cref{lem:local_sparse}, we obtain a constant $\delta \in (0,1)$, and let $\cE$ denote the event that the condition~\eqref{eq:locally_sparse} holds for $M^{(1)}, \dots, M^{(d)}$. 

We now condition on a choice of $M^{(1)}, \dots, M^{(d)}$ satisfying~\eqref{eq:locally_sparse}. Recall that for each hyperedge $\bfv = (v_{1}, v_{2}, v_{3}) \in \bigcup_{i=1}^{d} M^{(i)}$, the instance $I_{\yes}$ contains a constraint $(\bfv, S^{\mathsf{G}}_{b})$ and $I_{\no}$ contains a constraint $(\bfv, S^{\mathsf{G}}_{b'})$, for possibly distinct ``right-hand side'' coefficients $b, b' \in \mathsf{G}$. The right-hand sides of the equations in $I_{\no}$ are chosen independently, whereas those in $I_{\yes}$ are correlated through an underlying random assignment $\tau : [n] \times [3] \to \mathsf{G}$. 

However, it follows from \Cref{prop:locally_sparse} that for any variable set $U\subseteq [n]\times [3]$ of size at most $\delta n$, the equations placed by $I_{\yes}$ within $U$ have right-hand sides that remain jointly independent. Consequently, conditioned on $M^{(1)}, \dots, M^{(d)}$, the restricted random instances $\Sub_{U}[I_{\yes}]$ and $\Sub_{U}[I_{\no}]$ are identically distributed, as desired.
\end{proof}

\subsection{Concentration of Value for Random Instances}

The proofs of \Cref{lem:3LIN_exists_d,lem:expander_gadget} both hinge on the following concentration inequality, known as the Azuma-Hoeffding inequality.

\begin{proposition}[{\cite[Corollary 3.9]{Han14}}]\label{prop:Azuma}
Let $\{\mathscr{F}_{t}\}_{0\leq t\leq n}$ be a filtration of $\sigma$-algebras on a probability space, and let $X$ be a real-valued random variable. If there exists a constant $c>0$ such that $\big|\Ex{X|\mathscr{F}_{t}}-\Ex{X|\mathscr{F}_{t-1}}\big|\leq c$ almost surely for all $t\in [n]$, then 
\[
\Pr{\Ex{X|\mathscr{F}_{n}}-\Ex{X|\mathscr{F}_{0}}\geq \Delta\Big.}\leq\exp\left(-\frac{\Delta^{2}}{2nc^{2}}\right). 
\]
\end{proposition}

\begin{proof}[Proof Sketch of \Cref{lem:3LIN_exists_d}]
Recall from \Cref{def:Dyes} that, in process of generating the instance $I\sim \Dno^{\mathsf{G},n,d}$, the total of $dn$ constraints are placed in $d$ rounds. In each round, for each $j\in [n]$, exactly one constraint involving $(j,1)$ is added. This naturally divides the stochastic process of generating $I$ into $dn$ time steps, where exactly one constraint is added to $I$ in each step. Let $\mathscr{F}_{t}$ be the $\sigma$-algebra corresponding to the information after $t$ time steps are completed, for each $t\in \{0,1,\dots,dn\}$. For each fixed assignment $\tau:[n]\times [3]\rightarrow \mathsf{G}$, we define the random variable $X_{\tau}:=\val_{I}(\tau)$ and apply \Cref{prop:Azuma} to $X_{\tau}$ and the filtration $(\mathscr{F}_{t})_{0\leq t\leq dn}$. Using a ``switching'' argument (e.g., see \cite[Lemma 2.19]{wormald1999models}), it is not hard to see that 
\[\big|\Ex{X_{\tau}|\mathscr{F}_{t}}-\Ex{X_{\tau}|\mathscr{F}_{t-1}}\big|\leq \frac{2}{dn},\]
for any $t\in [dn]$. Furthermore, we have $\Ex{X_{\tau}|\mathscr{F}_{dn}}=X_{\tau}=\val_{I}(\tau)$ and $\Ex{X|\mathscr{F}_{0}}=\Ex{X_{\tau}}=|\mathsf{G}|^{-1}$. Therefore, the conclusion of \Cref{prop:Azuma} implies
\[
\Pru{I}{\val_{I}(\tau)\geq \frac{1}{|\mathsf{G}|}+\varepsilon
}\leq \exp\left(-\frac{\varepsilon^{2}dn}{8}\right).
\]
A union bound over all $\tau$ then implies
\[
\Pru{I}{\val_{I}\geq \frac{1}{|\mathsf{G}|}+\varepsilon
}\leq |\mathsf{G}|^{3n}\exp\left(-\frac{\varepsilon^{2}dn}{8}\right)\leq o(1),
\]
if we set $d\geq  100\varepsilon^{-2}\ln |\mathsf{G}|$.
\end{proof}

We now turn to prove \Cref{lem:expander_gadget}. The idea of the proof is similar to the above proof of \Cref{lem:3LIN_exists_d}, but more preparation is needed before starting the proof. 

\begin{definition}
For any nonempty finite set $D$, we define a relation $\Perm_{D}:D^{D}\rightarrow \{0,1\}$ as follows. For any tuple $y=(y_{x})_{x\in D}\in D^{D}$, we let $\Perm_{D}(y)=1$ if and only if the map $x\mapsto y_{x}$ is a permutation of $D$.
\end{definition}

Note that $\Perm_{D}$ is the ``maximal'' sub-unique relation on $D$, as $R(y)\leq \Perm_{D}(y)$ for any sub-unique relation $y:D^{D}\rightarrow \{0,1\}$.

\begin{definition}\label{def:perm_value}
For any assignment $\tau:[n]\times D\rightarrow D$, we define its \emph{Perm-value} to be
\[
\textup{Perm-val}(\tau):=\Exu{\bfv=(v_{x})_{x\in D}}{\Perm_{D}\big((\tau(v_{x}))_{x\in D}\big)},
\]
where $\bfv=(v_{x})_{x\in D}$ is a random hyperedge chosen by sampling each $v_{x}$ from $[n]\times \{x\}$ independently and uniformly at random.
\end{definition}

\begin{lemma}\label{lem:perm-align}
For any assignment $\tau:[n]\times D\rightarrow D$ such that $\textup{Perm-val}(\tau)\geq1-\frac{1}{5|D|}$, there exists a permutation $\pi:D\rightarrow D$ such that
\begin{equation}\label{eq:perm-align}
\Pru{i\in [n]}{\tau\big((i,x)\big)= \pi(x)}\geq \textup{Perm-val}(\tau)
\end{equation}
holds for all $x\in D$.
\end{lemma}
\begin{proof}
We first show the existence of a map $\pi:D\to D$ satisfying \eqref{eq:perm-align} for all $x\in D$, and subsequently prove that the map $\pi$ thus obtained is a permutation.

For any $x,z\in D$, let \(
F(x,z)
\)
denote the fraction of indices $i\in [n]$ such that $\tau((i,x))=z$.  
Fix an element $t\in D$ and consider a hyperedge $\bfv=(v_x)_{x\in D}$. We decompose it into $v_t$ and $\bfv_{\setminus\{t\}}=(v_x)_{x\in D\setminus\{t\}}$.  
For the map $x\mapsto \tau(v_x)$ to be a permutation of $D$, two conditions must hold:
\begin{enumerate}[label=(\arabic*)]
    \item The values $\tau(v_x)$ for $x\in D\setminus\{t\}$ are pairwise distinct.
    \item The value $\tau(v_t)$ is the unique remaining element of $D$.
\end{enumerate}
With $\tau$ fixed and $\bfv$ chosen uniformly at random, conditioning on the requirement (1), the probability that (2) also holds is at most $\max_{z\in D} F(t,z)$. Consequently, the overall probability that $(\tau(v_x))_{x\in D}$ forms a permutation (i.e. $\textup{Perm-val}(\tau)$) is at most $\max_{z\in D} F(t,z)$. Define $\pi(t)$ as the element of $D$ attaining this maximum, i.e.
\[
F(t,\pi(t)) = \max_{z\in D} F(t,z).
\]
By the discussion above, for any $t\in D$ we have
\[
\Pru{i\in [n]}{\tau\big((i,t)\big)=\pi(t)}=\max_{z\in D}F(t,z)\geq \textup{Perm-val}(\tau).
\]

It remains to show that $\pi$ is a permutation. For a uniformly random hyperedge $\bfv$, we have
\[
\Pru{\bfv=(v_{x})_{x\in D}}{\tau(v_{x})=\pi(x),\quad\forall\, x\in D\big.}\geq \prod_{x\in D} F(x,\pi(x))\geq \textup{Perm-val}(\tau)^{|D|}\geq \left(1-\frac{1}{5|D|}\right)^{|D|}.
\]
If $\pi$ were not a permutation, then
\[
\textup{Perm-val}(\tau)=\Exu{\bfv=(v_{x})_{x\in D}}{\Perm_{D}\big((\tau(v_{x}))_{x\in D}\big)}\leq 1-\left(1-\frac{1}{5|D|}\right)^{|D|}< 1-\frac{1}{10|D|},
\]
contradicting the assumption. Hence, $\pi$ must be a permutation.
\end{proof}

We are now ready to prove \Cref{lem:expander_gadget}.

\begin{proof}[Proof of \Cref{lem:expander_gadget}]
We call an assignment $\tau:[n]\times D\rightarrow D$ ``bad'' (with respect to $H$) if $\val_{\cI[R,H]}(\tau)\geq 1-\varepsilon$ but 
\begin{equation}\label{eq:perm_assumption}
\Pru{(i,x)\in [n]\times D}{\tau\big((i,x)\big)\neq y_{x}}> 2\varepsilon
\end{equation}
for any tuple $y\in D^{D}$ with $R(y)=1$. For any fixed assignment $\tau:[n]\times D\rightarrow D$, we will analyze the probability that it is bad with respect to a random hypergraph $H\sim \cH_{D,n,\ell}$.

\textbf{Case 1: $\textnormal{Perm-val}(\tau)\geq 1-2\varepsilon$.} Since $2\varepsilon\leq \frac{1}{5|D|}$, we may apply \Cref{lem:perm-align} and obtain a permutation $\pi:D\rightarrow D$ such that 
\[
\Pru{i\in [n]}{\tau\big((i,x)\big)=\pi(x)}\geq 1-2\varepsilon
\]
holds for any $x\in D$. If $R\big((\pi(x))_{x\in D}\big)=1$, then this contradicts the assumption that \eqref{eq:perm_assumption} holds for all $y\in D^{D}$ with $R(y)=1$. Therefore, we must have $R\big((\pi(x))_{x\in D}\big)=0$.

Now if $\bfv=(v_{x})_{x\in D}$ is a random hyperedge chosen by sampling each $v_{x}$ from $[n]\times \{x\}$ independently and uniformly at random, then
\[
\Exu{\bfv=(v_{x})_{x\in D}}{R\big((\tau(v_{x}))_{x\in D}\big)}\leq 1-\Pru{\bfv=(v_{x})_{x\in D}}{\tau(v_{x})=\pi(x),\quad\forall\, x\in D\big.}\leq 1-(1-2\varepsilon)^{|D|}\leq 1-2\varepsilon.
\]
Using the martingale argument in the proof of \Cref{lem:3LIN_exists_d}, it follows that
\[\Pru{H}{\val_{\cI[R,H]}(\tau)\geq 1-\varepsilon}\leq \exp\left(-\frac{\varepsilon^{2}\ell n}{8}\right).\]
Therefore, the probability that $\tau$ is bad with respect to a random $H$ is at most $\exp(-\varepsilon^{2}\ell n/8)$.

\textbf{Case 2: $\textnormal{Perm-val}(\tau)<1-2\varepsilon$.} Since $R(y)\leq \Perm_{D}(y)$ for all $y\in D^{D}$, by \Cref{def:perm_value} we have
\[
\Pru{\bfv=(v_{x})_{x\in D}}{R\big(\tau(v_{x})_{x\in D}\big)}\leq \Pru{\bfv=(v_{x})_{x\in D}}{\Perm_{D}\big(\tau(v_{x})_{x\in D}\big)}=\textup{Perm-val}(\tau)<1-2\varepsilon.
\]
As in case 1, it follows again that $\tau$ is bad with respect to $H$ with probability at most $\exp(-\varepsilon^{2}\ell n/8)$.

In conclusion, any fixed assignment $\tau:[n]\times D\rightarrow D$ is bad with respect to $H$ with probability at most $\exp(-\varepsilon^{2}\ell n/8)$. By a union bound over all assignments $\tau$, we have
\[
\Pru{H}{\text{Some assignment }\tau\text{ is bad with respect to }H\big.}\leq |D|^{n|D|}\exp\left(-\frac{\varepsilon^{2}\ell n}{8}\right)=\exp(-\Omega(n)),
\]
as long as $\ell\geq 100\varepsilon^{-2}|D|\ln |D|$.
\end{proof}
\section{Varieties Admitting Affine-type}\label{sec:varieties}

In this section, we briefly discuss the background of \Cref{thm:main_algebra}.

\subsection{Varieties}

Recall from \Cref{def:algebraic_homomorphism} that the notion of an algebraic homomorphism is well-defined only when the operations of the source and target structures are placed in a bijective correspondence, and each pair of corresponding operations has the same arity. Such pairs of algebraic structures are said to share the same \emph{signature}.

In universal algebra, it is often convenient to study entire classes of algebraic structures that share a common signature. For such a class, one can represent all operation sets uniformly by a single (possibly infinite) collection of \emph{operation symbols}. Each operation symbol has a fixed finite arity but may represent different concrete operations in different structures. Even so, it is meaningful to consider \emph{identities} among these operation symbols on a purely syntactic level. For instance, if $\cS$ is a set of operation symbols and $s_{1}, s_{2} \in \cS$, then
\[
s_{1}(x_{1}, s_{2}(x_{2}, x_{3})) = s_{2}(s_{1}(x_{1}, x_{2}), x_{3}) \quad \forall\, x_{1}, x_{2}, x_{3}
\]
is an example of an identity, which may or may not hold in a particular algebraic structure interpreting the symbols in~$\cS$.

\begin{definition}
Let $\cS$ be a (possibly infinite) set of operation symbols, each with a finite arity.  
A \emph{term} is any finite expression built from variables $x_{1}, x_{2}, \dots$ and the operation symbols in~$\cS$.  
If $t_{1}$ and $t_{2}$ are terms, a syntactic statement of the form ``$t_{1} = t_{2}$ for all $x_{1}, x_{2}, \dots$’’ is called an \emph{identity}.  
Given a (possibly infinite) set of identities $\cT$, the \emph{variety} $\cV(\cS; \cT)$ is the class of all algebraic structures interpreting $\cS$ such that every identity in~$\cT$ holds.
\end{definition}

\begin{definition}
For any algebraic structure $\bA$, we define the variety generated by $\bA$ to be $\cV(\cS;\cT)$, where $\cS$ is the set of operation symbols interpreted by $\bA$ and $\cT$ is the set of identities in $\cS$ that hold for $\bA$. The variety generated by $\bA$ is also denoted by $\cV(\bA)$.
\end{definition}

For example, it is not hard to see that any subalgebra or homomorphic image (defined in \Cref{def:subalgebra}) of $\bA$ is contained in $\cV(\bA)$. 

\subsection{Types of Strictly Simple Algebras}

Following conventions in universal algebra, we refer to an algebraic structure simply as an ``algebra.'' An algebra of $k$ elements is said to be \emph{simple} if $k\geq 2$ and every homomorphic image of it has either 1 element or $k$ elements. It is further said to be \emph{strictly simple} if every subalgebra of it has either 1 element or $k$ elements. Tame Congruence Theory, developed by Hobby and Mckenzie \cite{hobby1988structure}, classified strictly simple algebras into five types: (1) the unary type, (2) the affine type, (3) the Boolean type, (4) the lattice type and (5) the semilattice type. Later, Szendrei \cite{szendrei1992survey} provided a characterization of strictly simple \emph{idempotent} algebras of all five types. We record the following immediate corollary of \cite[Theorem 6.1]{szendrei1992survey}; readers interested in further details are referred to \cite{hobby1988structure,szendrei1992survey}.

\begin{proposition}[\cite{szendrei1992survey}]
If $\bA$ is a strictly simple idempotent algebra of unary type or affine type, then there exists an Abelian group structure on the underlying set of $\bA$ such that for any operation $f$ in $\bA$ of arity $k$, the equation 
\[
f(x_{1},\dots,x_{k})-f(y_{1},\dots,y_{k})+f(z_{1},\dots,z_{k})=f(x_{1}-y_{1}+z_{1},\dots,x_{k}-y_{k}+z_{k})
\]
holds for all $x_{1},\dots,x_{k},y_{1},\dots,y_{k},z_{1},\dots,z_{k}\in \bA$.
\end{proposition}
Using the idempotency condition, it is easy to deduce the following corollary.

\begin{corollary}\label{cor:Szendrei}
Any strictly simple idempotent algebra of unary type or affine type has an Abelian group $\mathsf{G}$ as its underlying set, and all of its operations belong to $\Pol(\TLin_{\mathsf{G}})$.
\end{corollary}

In general, a finite algebra (not necessarily strictly simple) is associated with a subset of the five possible types \{unary, affine, Boolean, lattice, semilattice\}. A variety $\cV$ is said to \emph{admit} a type if that type appears in the type set of at least one algebra belonging to $\cV$. The following theorem by Valeriote \cite{valeriote2009subalgebra} is a crucial component of \Cref{thm:main_algebra}.

\begin{theorem}[{\cite[Proposition~3.1]{valeriote2009subalgebra}}]\label{thm:Valeriote}
Let $\bA$ be a finite idempotent algebra. If the variety $\cV(\bA)$ admits either the unary type or the affine type, then some homomorphic image of a subalgebra of $\bA$ is a strictly simple algebra of unary or affine type.
\end{theorem}

Note that homomorphic images and subalgebras of an idempotent algebra are also automatically idempotent.

The final ingredient in \Cref{thm:main_algebra} is the following characterization of bounded-width CSP templates in terms of the variety generated by its polymorphism algebra. It was conjectured by \cite{larose2007bounded} and proved by \cite{barto2014constraint}.

\begin{theorem}[{\cite[Conjecture 4.3]{barto2014constraint}}]\label{thm:BK14}
Let $(D,\Gamma)$ be a relational structure that is a core. Then $(D,\Gamma)$ has unbounded width if and only if the variety generated by the algebra $(D,\Pol(\Gamma))$ admits the unary type or the affine type.
\end{theorem}

It is then clear that \Cref{thm:main_algebra} follows from the combination of \Cref{cor:Szendrei} and \Cref{thm:Valeriote,thm:BK14}.

\end{document}